\documentclass{llncs}

\renewcommand{\subparagraph}{}
\usepackage[small,compact]{titlesec}
\input{Userdef.sty}

\begin{document}

\title{Abstraction of Elementary Hybrid Systems by Variable Transformation}


\author{Jiang Liu\inst{1} \and Naijun Zhan\inst{2} \and Hengjun Zhao \inst{1}\and Liang Zou \inst{2}}
\institute{Chongqing Key Lab. of Automated Reasoning and Cognition, CIGIT, CAS \and
  State Key Lab. of Comput. Sci.,
  Institute of Software, CAS
}

\maketitle

\begin{abstract}
Elementary hybrid systems (EHSs) are those hybrid systems (HSs) containing elementary functions such as {\scriptsize \sf exp, ln, sin, cos}, etc.
EHSs are very common in practice, especially in safety-critical domains. Due to the non-polynomial expressions which lead to undecidable arithmetic, verification of EHSs is very hard. Existing approaches based on partition of state space or over-approximation of reachable sets suffer from state explosion or inflation of numerical errors. In this paper, we propose a symbolic abstraction approach that reduces EHSs to polynomial hybrid systems (PHSs), by replacing all non-polynomial terms with newly introduced variables. Thus the verification of EHSs is reduced to the one of PHSs, enabling us to apply all the well-established verification techniques and tools for PHSs to EHSs. In this way, it is possible to avoid the limitations of many existing methods. We illustrate the abstraction approach and its application in safety verification of EHSs by several real world examples.
\end{abstract}


\section{Introduction}\label{sec:intro}
Complex Embedded Systems (CESs) consist of software and hardware components that operate autonomous devices interacting with the physical environment. They are now part of our daily life and are used in many industrial sectors 
to carry out highly complex and often critical functions.
The development process of CESs is widely recognized as a highly complex and challenging task. A thorough validation
and verification activity is necessary to enhance the quality of CESs and, in particular, to fulfill the quality criteria mandated by the relevant standards. Hybrid systems (HSs) are mathematical models with precise mathematical semantics for CESs, wherein continuous physical dynamics are combined with discrete transitions. Based on HSs, rigorous analysis and
verification of CESs become feasible, so that errors can be detected and corrected in the very early stage of design.

In practice, it is very common to model complex physical environments by ordinary differential equations (ODEs) with elementary functions such as reciprocal function $\frac{1}{x}$, exponential function $e^x$, logarithm function $\ln x$, trigonometric functions $\sin x$ and $\cos x$, and their compositions. We call such HSs elementary HSs (EHSs).
As elementary expressions usually lead to undecidable arithmetic, the verification of EHSs becomes very hard, even intractable.
Existing methods that deal with EHS verification include the level-set method \cite{MT00}, the hybridization method \cite{ETA07,JGMD12}, the gridding-based abstraction refinement method \cite{RS07}, the interval SMT solver-based method \cite{GKC13-dreach,iSAT-ODE}, the Taylor model-based flowpipe approximation method \cite{CAS12}, and so on. These methods rely either on iterative partition of state space or on iterative computation of approximate reachable sets, which can quickly lead to explosion of state numbers or inflation of numerical errors.
Moreover, most of the above mentioned methods can only do bounded model checking (BMC).

As an alternative, the constraint-based approach verifies the safety property of a HS by solving corresponding constraints symbolically or numerically, to discover a barrier (inductive invariant) that separates the reachable set from the unsafe region, which avoids exhaustive gridding or brute-force computation, and can thus overcome the limitations of the above mentioned methods. However, this method has mainly been applied to verification of polynomial hybrid systems (PHSs) \cite{SSM04a,PJP07,PC08,GT08,LZZ11}.
Although ideas about generating invariants for EHSs appeared in \cite{PC08,GP14}, they were talked about in an ad hoc way.
In \cite{Sankaranarayanan12}, the author proposed a change-of-bases method to transform EHSs to PHSs, even to linear systems, but the success depends on the choice of the set of basis functions, and therefore does not apply to general EHSs.

In this paper, we investigate symbolic abstraction of general EHSs to PHSs, by extending \cite{Sankaranarayanan12} with early works on polynomilization of elementary ODEs \cite{Kerner81,SV87}. Herein the abstraction is accomplished by introducing new variables to replace the non-polynomial terms. With the substitution, flows, guards and other components of the EHSs are transformed according to the chain rule of differentiation, or by the over-approximation methods proposed in the paper, so that for any trajectory of the EHSs, there always exists a corresponding trajectory of the reduced PHSs. Besides, such abstraction preserves (inductive) invariant sets. Therefore, verification of the EHSs is naturally reduced to the one of the reduced PHSs. This will be shown by several real world verification problems.

The proposed abstraction applies to general EHSs.
The benefit of the proposed abstraction is that it enables all the well-established verification techniques and tools for PHSs, especially the constraint-based approaches such as DAL \cite{Platzer10} and SOS \cite{PJP07,KHSH13}, to be applied to EHSs, and thus provides the possibility of avoiding such limitations as error inflation, state explosion and boundedness for existing EHS verification methods.
A by-product is that it also provides the possibility of generating invariants with elementary functions for PHSs, thus enhancing the power
of existing PHS verification methods.
In short, the proposed abstraction method can be a good alternative or complement to existing approaches.


\paragraph*{Related Work.} This work is most closely related to \cite{Sankaranarayanan12} and \cite{Kerner81}. The abstraction in this paper is performed by systematic augmentation of the original system rather than change-of-bases, thus essentially different from \cite{Sankaranarayanan12} and generally applicable. Compared to \cite{Kerner81}, this paper gives a clearer reduction procedure for elementary ODEs and discusses the extension to hybrid systems. This work is most closely related to \cite{Sankaranarayanan12} and \cite{Kerner81}. The abstraction in this paper is performed by systematic augmentation of the original system rather than change-of-bases, thus essentially different from \cite{Sankaranarayanan12} and more general. Compared to \cite{Kerner81}, this paper gives a clearer reduction procedure for elementary ODEs and discusses the extension to hybrid systems.
It was proved in \cite{Ratschan14} that safety verification of nonlinear hybrid systems is quasi-semidecidable, but to find efficient verification algorithms remains an open problem. An approximation technique for abstracting nonlinear hybrid systems to PHSs based on Taylor polynomial was proposed in \cite{LT07}, but to abstract the continuous flow transitions it requires the ODEs to have closed-form solutions. In \cite{PP05}, the authors adopted similar recasting techniques to ours for stability analysis of non-polynomial systems. Regarding non-polynomial invariants for polynomial continuous or hybrid systems, \cite{RMM12} presented the first method for generating transcendental invariants using formal power series, while the more recent work \cite{GJPS14} proposed a Darboux Polynomial-based method. Both \cite{RMM12} and \cite{GJPS14} can only find non-polynomial invariants of limited forms.

\paragraph*{Paper Organization.} The rest of the paper is organized as follows. We briefly review some basic notions about hybrid systems and the theory of abstraction for hybrid systems in Section~\ref{sec:preli}. Section \ref{sec:abstraction} is devoted to the transformation from EDSs to PDSs, and from EHSs to PHSs. Section \ref{sec:application} discusses how to use the proposed abstraction approach for safety verification of EHSs. Section \ref{sec:conclusion} concludes this paper.

\section{Preliminary}\label{sec:preli}
In this section, we briefly introduce the basic knowledge of hybrid systems and define what we call \emph{elementary hybrid systems}. Besides, we also recall the basic theory of abstraction for hybrid systems originally developed in \cite{Sankaranarayanan11,Sankaranarayanan12}.

Throughout this paper, we use $\mathbb N,\mathbb Q,\mathbb R$ to denote the set of \emph{natural}, \emph{rational} and \emph{real} numbers respectively. Given a set $A$, the power set of $A$ is denoted by $2^A$, and the Cartesian product of $n$ duplicates of $A$ is denoted by $A^n$; for instance, $\mathbb R^n$ stands for the $n$-dimensional Euclidean space. A vector element $(a_1,a_2,\ldots,a_n)\in A^n$ is usually abbreviated by a boldface letter $\mathbf a$ when its dimension is clear from the context.

\subsection{Elementary Continuous and Hybrid Systems}
A continuous dynamical system (CDS) is modeled by first-order autonomous ordinary differential equations (ODEs)
\begin{equation}\label{eqn:ode}
  \dot \xx= \fb(\xx),
\end{equation}
where $\xx=(x_1,\ldots,x_n)\in\mathbb{R}^n$ and $\fb: U\rightarrow \mathbb R^n$ is a vector function, called a vector field, defined on an open set $U\subseteq \mathbb R^n$.
If $\fb$ satisfies the \emph{local Lipschitz condition} \cite{Khalil01}, then for any $\xx_0\in U$, there exists a unique {differentiable} vector function $\xx(t): (a,b)\rightarrow U$, where $(a, b)$ is an open interval containing $0$, such that $\xx(0)=\xx_0$ and the derivative of $\xx(t)$ w.r.t. $t$ satisfies $\forall t\in (a,b).\,{\ud \xx(t)\over \ud t} = \fb (\xx(t))$. Such $\xx(t)$ is called the \emph{solution} to (\ref{eqn:ode}) with initial value $\xx_0$, or the \emph{trajectory} of (\ref{eqn:ode}) starting from $\xx_0$.

In many contexts, a CDS $\mathcal{C}$ may be equipped with an initial set $\Xi$ and a domain ${D}$, represented as a triple $\mathcal C\,\define\,(\Xi, \fb, {D})$.\footnote{In this paper, the symbol $\,\define\,$ is interpreted as ``defined as".} If $\fb$ is defined on $U\subseteq \mathbb R^n$, then $\Xi$ and $D$ should satisfy $\Xi\subseteq D\subseteq U$. In what follows, all CDSs will refer to the triple form unless otherwise stated.
Hybrid systems (HSs) are those systems that exhibit both continuous evolutions and discrete transitions. A
popular model of HSs is \emph{hybrid automata} \cite{ACHH93,Henzinger96}.

\begin{definition}[Hybrid Automaton]\label{dfn:HS}
A {hybrid automaton} (HA) is a system $\mathcal{H} \, \define\, (Q, X, f,$
$D, E, G, R, \Xi)$,\, where
\begin{itemize}
  \item $Q=\{q_1,\ldots,q_m\}$ is a finite set of modes;
  \item $X=\{x_1,\ldots,x_n\}$ is a finite set of continuous state variables, with $\xx=(x_1,\ldots,x_n)$ ranging over $\mathbb R^n$;
  \item $f: Q\rightarrow (U_q\rightarrow \mathbb R^n)$ assigns to each mode $q\in Q$ a locally Lipschitz continuous vector field $\fb_q$ defined on the open set $U_q\subseteq \mathbb R^n$;
  \item $D$ assigns to each mode $q\in Q$ a domain $D_q\subseteq U_q$;
  \item $E\subseteq Q\times Q$ is a finite set of discrete transitions;
  \item $G$ assigns to each transition $e\in E$ a guard $G_e\subseteq \mathbb R^n$;
  \item $R$ assigns to each  transition $e\in E$ a set-valued reset function $R_e$: $G_e\rightarrow 2^{\mathbb R^n}$;
  \item $\Xi$ assigns to each $q\in Q$ a set of initial states $\Xi_q\subseteq D_q$.
\end{itemize}
\end{definition}

Actually a HA can be regarded as a composition of a finite set of CDSs $\mathcal C_q\,\define\,(\Xi_q,\fb_q,$ $D_q)$ for $q\in Q$, together with the set of transition relations specified by $(G_e,R_e)$ for $e\in E$. Conversely, any CDS can be regarded as a special HA with a single mode and without discrete transitions.

In this paper, we consider the class of HSs that can be defined by
multivariate \emph{elementary} functions given by the following grammar:
\begin{eqnarray}
  f,g &::=& c\mid x \mid f+g \mid f-g\mid f \times g\mid \label{eqn:ele-poly} \\
  & & \frac{f}{g}\mid f^{a}\mid e^f\mid \ln(f)\mid \sin(f)\mid \cos(f) \enspace ,\label{eqn:ele-trans}
\end{eqnarray}
where $c\in\mathbb R$ is any real constant, $a\in\mathbb Q$ is any rational constant, and $x$ can be any variable from the set of real-valued variables $\{x_1, \ldots$, $x_n\}$.
In particular, the set of functions constructed only by (\ref{eqn:ele-poly}) are multivariate {\emph{polynomials}} in $x_1,x_2,\ldots,x_n$.

\begin{definition}[Elementary and Polynomial HSs]\label{dfn:ele-hs}
A HS or CDS is called \emph{elementary} (resp. \emph{polynomial})
if it can be expressed by \emph{elementary} (resp. \emph{polynomial}) functions together with relational symbols $\geq,>,\leq,<,=,\neq$ and Boolean connectives $\wedge,\vee,\neg,$ $\longrightarrow,\longleftrightarrow$.
\end{definition}

Elementary (resp. {polynomial}) HSs or CDSs will be denoted by EHSs or EDSs (resp. PHSs or PDSs) for short.

\begin{remark}
 The limitation of elementary functions to grammar (\ref{eqn:ele-poly}) and (\ref{eqn:ele-trans}) is not essential. For example, tangent and cotangent functions $\tan(f),\cot(f)$ can be easily defined. Besides, the presented approach in this paper is also applicable to other elementary functions not mentioned above, such as inverse trigonometric functions $\arcsin(f),$ $\arccos(f)$, etc. However, it does exclude functions like
 \begin{displaymath}
f(x) = \left\{ \begin{array}{ll}
\frac{\sin x}{x} & \textrm{if $x\neq 0$}\\
\,\,\,\,\scriptstyle{1} & \textrm{if $x=0$}\\
\end{array} \right..
\end{displaymath}
\end{remark}

\subsection{Semantics of Hybrid Systems}
Given a HA $\mathcal H$, denote the state space of $\mathcal H$ by $\mathbb H\,\define\,Q\times\mathbb R^n$, the domain of $\mathcal H$ by $D_{\mathcal H}\,\define\,\bigcup_{q\in Q} (\{q\}\times {D}_q)$, and the set of all initial states by $\Xi_{\mathcal H}\,\define\,\bigcup_{q\in Q} (\{q\}\times \Xi_q)$.
The semantics of $\mathcal H$ can be characterized by the set of \emph{reachable} states of $\mathcal H$.

\begin{definition}[Reachable Set] \label{dfn:RS}
Given a HA $\mathcal H$, the \emph{reachable set} of $\mathcal H$, denoted by $\mathcal R_{\mathcal H}$, consists of such $(q,\xx)\in \mathbb H$\, for which there exists a finite sequence
$$(q_0,\xx_0),(q_1,\xx_1),\ldots,(q_l,\xx_l)$$
such that $(q_0,\xx_0)\in \Xi_{\mathcal H}$, $(q_l,\xx_l)=(q,\xx)$, and for any $0\leq i\leq l-1$, one of the following two conditions holds:
\begin{itemize}
  \item (Discrete Jump): $e=(q_i,q_{i+1})\in E$, \,$\xx_i\in G_e$ and $\xx_{i+1}\in R_e(\xx_i)$; or
  \item (Continuous Evolution): $q_i=q_{i+1}$, and there exists a $\delta\geq 0$ s.t. the trajectory $\xx(t)$ of $\dot \xx=\fb_{q_i}$ starting from $\xx_i$ satisfies
      \begin{itemize}
         \item $\xx(t)\in D_{q_i}$ for all $t\in [0,\delta]$; and
         \item $\xx(\delta)=\xx_{i+1}$\,.
      \end{itemize}
\end{itemize}
\end{definition}

Exact computation of reachable sets of hybrid systems is generally an intractable problem. For verification of safety properties, appropriate over-approximations of reachable sets will suffice.
\begin{definition}[Invariant]\label{dfn:inv}
Given a HA $\mathcal H$, a set $\mathcal I \,\define\,\bigcup_{q\in Q} (\{q\}\times I_q)\subseteq \mathbb H$ is called an \emph{invariant} of $\mathcal H$, if $\mathcal I$ is a superset of the reachable set $\mathcal R_{\mathcal H}$, i.e. $\mathcal R_{\mathcal H}\subseteq \mathcal I$.
\end{definition}

\begin{definition}[Inductive Invariant]\label{dfn:induc-inv}
Given a HA $\mathcal H$, a set $\mathcal I \,\define\,\bigcup_{q\in Q} (\{q\}\times I_q)\subseteq \mathbb H$ is called an \emph{inductive invariant} of $\mathcal H$, if $\mathcal I$ satisfies the following conditions:
\begin{itemize}
  \item $\Xi_q\subseteq I_{q}$ for all $q\in Q$;
  \item for any $e=(q,q')\in E$, if $\xx\in I_{q}\cap G_e$, then $R_e(\xx)\subseteq I_{q'}$;
  \item for any $q\in Q$ and any $\xx_0\in I_q$, if ${\xx}(t)$ is the trajectory of $\dot \xx=\fb_q$ starting from $\xx_0$, and there exists $T\geq 0$ s.t. ${\xx}(t)\in D_q$ for all $t\in [0,T]$, then ${\xx}(T)\in I_q$\,.
\end{itemize}
\end{definition}

It is easy to check that any inductive invariant is also an invariant.
%

\subsection{Abstraction of Hybrid Systems}
We next briefly introduce the kind of abstraction for HSs proposed in \cite{Sankaranarayanan11,Sankaranarayanan12} and the significant properties about such abstraction.

In what follows, to distinguish between the dimensions of a HS and its abstraction, we will annotate a HS $\mathcal{H}$ (a CDS $\mathcal{C}$) with the vector of its continuous state variables $\xx$ as $\mathcal{H}_{\xx}$ ($\mathcal{C}_{\xx}$). We use $|\xx|$ to denote the dimension of $\xx$.
Given a vector function $\Theta$ that maps from $D\subseteq \RR^{|\xx|}$ to $\RR^{|\yy|}$, let $\Theta(A)\,\define\,\{\Theta(\xx)\mid \xx \in A\}$ for any $A\subseteq D$, and $\Theta^{-1}(B)\,\define\,\{\xx\in D\mid \Theta(\xx)\in B \}$ for any $B\subseteq \mathbb R^{|\yy|}$.

\begin{definition}[Simulation \cite{Sankaranarayanan11}]\label{dfn:simu-cds}
Given two CDSs $\mathcal{C}_{\xx}\,\define\,$ $(\Xi_{\xx}, \fb_{\xx},D_{\xx})$ and $\mathcal{C}_{\yy} \,\define\, (\Xi_{\yy},\fb_{\yy}, $ $ D_{\yy})$, we say
$\mathcal C_{\yy}$ {\it simulates} $\mathcal C_{\xx}$ or $\mathcal C_{\xx}$ is simulated by $\mathcal C_{\yy}$ via a continuously differentiable mapping $\Theta:D_{\xx} \rightarrow \RR^{|\yy|}$, if $\Theta$ satisfies
\begin{itemize}
\item $\Theta(\Xi_{\xx})\subseteq \Xi_{\yy}$, $\Theta(D_{\xx})\subseteq D_{\yy}$; and
\item for any trajectory $\xx(t)$ of $\mathcal C_{\xx}$ (i.e. a trajectory of $\dot \xx = \fb_{\xx}(\xx)$ that starts from  $\Xi_{\xx}$ and stays in $D_{\xx}$), \,$\Theta\circ\xx(t)$ is a trajectory of $\mathcal C_{\yy}$, where $\circ$ denotes composition of functions.
\end{itemize}
We call $\mathcal C_{\yy}$ an {\em abstraction} of $\mathcal C_{\xx}$ under the {\em simulation map} $\Theta$.
\end{definition}

Abstraction of a HS can be obtained by abstracting the CDS corresponding to each mode using an individual simulation map. As argued in \cite{Sankaranarayanan12}, it can be assumed without loss of generality that the collection of simulation maps for each mode all map to an Euclidean space of the same dimension, say $\RR^{|\yy|}$.

\begin{definition}[Simulation \cite{Sankaranarayanan12}]\label{dfn:simu-HS}
Given two HSs $\mathcal{H}_{\xx}\,\define\, (Q,X,f_{\xx},D_{\xx},E,G_{\xx},R_{\xx},\Xi_{\xx})$ and
$\mathcal{H}_{\yy}\,\define\,(Q,Y,f_{\yy},D_{\yy},E,G_{\yy},$ $R_{\yy},\Xi_{\yy})$, we say $\mathcal H_{\yy}$ simulates $\mathcal H_{\xx}$ via the set of maps $\{\Theta_q: D_{\xx,q}\rightarrow \RR^{|\yy|}\mid q\in Q\}$, if the following hold:
\begin{itemize}
\item $(\Xi_{{\yy},q},\fb_{{\yy},q}, D_{{\yy},q})$ simulates $(\Xi_{{\xx},q}, \fb_{{\xx},q}, D_{{\xx},q})$ via $\Theta_q$, for each $q\in Q$;
\item $\Theta_q( G_{{\xx},e})\subseteq  G_{{\yy},e}$, for any $e=(q,q')\in E$;
\item $\Theta_{q'}( R_{{\xx},e}({\xx}))\subseteq R_{{\yy},e}(\Theta_q({\xx}))$, for any $e=(q,q')\in E$ and any ${\xx}\in G_{\xx,e}$.
\end{itemize}
We call $\mathcal H_{\yy}$ an {\em abstraction} of $\mathcal H_{\xx}$ under the set of {\em simulation maps} $\{\Theta_q\mid q\in Q\}$.
\end{definition}

Intuitively, if $\mathcal H_{\yy}$ is an abstraction of $\mathcal H_{\xx}$, then for any $(q,\xx)$ reachable by $\mathcal H_{\xx}$, $(q,\Theta_q(\xx))$ is a state reachable by  $\mathcal H_{\yy}$. Actually, we can prove the following nice property about such abstractions.

\begin{theorem}[Invariant Preserving Property]\label{thm:SimuInvHa}
If $\mathcal H_{\yy}$ is an {abstraction} of $\mathcal H_{\xx}$ under {simulation maps} $\{\Theta_q\mid q\in Q\}$, and $\mathcal I_{\yy}\,\define\, \bigcup_{q\in Q} (\{q\}\times I_{\yy,q})$ is an invariant (resp. inductive invariant) of $\mathcal{H}_{\yy}$, then
$\mathcal I_{\xx}\,\define\,\bigcup_{q\in Q} (\{q\}\times I_{\xx,q})$ with $I_{\xx,q}\,\define\,\Theta_{q}^{-1}(I_{\yy,q})$
is an invariant (resp. inductive invariant) of $\mathcal{H}_{\xx}$.
\end{theorem}

Theorem \ref{thm:SimuInvHa} extends Theorem 3.2 of \cite{Sankaranarayanan11} in two aspects: firstly, it deals with HSs, and secondly, it applies to both invariants and inductive invariants; nevertheless, the proof of Theorem  \ref{thm:SimuInvHa} can be given in a similar way and so is omitted here.
The significance of Theorem \ref{thm:SimuInvHa} lies in the possibility of analyzing a complex HS by analyzing certain abstractions of it, which may be of simpler forms and thus allow the use of any available techniques and tools.

The following theorem proposed in \cite{Sankaranarayanan11} is very useful for checking or constructing simulation maps.
\begin{theorem}[Simulation Checking \cite{Sankaranarayanan11}]\label{thm:sim-check}
 Let $\mathcal C_{\xx}$, $\mathcal C_{\yy}$,
 $\Theta$ be specified as in Definition \ref{dfn:simu-cds}. Suppose $|\xx|=n,|\yy|=\ell$, and  $\Theta\,\define\,(\theta_1,\theta_2,\ldots, \theta_{\ell})$.
 Then $\mathcal C_{\yy}$ simulates $\mathcal C_{\xx}$ if
 \begin{itemize}
   \item  $\Theta(\Xi_{\xx})\subseteq \Xi_{\yy}$, $\Theta(D_{\xx})\subseteq D_{\yy}$; and
    \item $\fb_{\yy}(\Theta(\xx))=\mathcal J_{\Theta}(\xx) \cdot \fb_{\xx}(\xx)$, for any $\xx\in D_{\xx}$, where $\fb_{\xx}(\xx)$ is seen as a column vector, and $\mathcal J_{\Theta}(\xx)$ represents the \emph{Jacobian matrix} of $\Theta$ at point $\xx$, i.e.
      \begin{displaymath}
      \mathcal J_{\Theta}(\xx)=\left(
        \begin{array}{ccc}
          \frac{\partial{\theta_1}}{\partial x_1} & \ldots &  \frac{\partial{\theta_1}}{\partial x_n} \\
          \vdots & \ddots & \vdots \\
          \frac{\partial{\theta_{\ell}}}{\partial x_1} & \ldots & \frac{\partial{\theta_{\ell}}}{\partial x_n}  \\
        \end{array}
      \right)\enspace.
      \end{displaymath}
 \end{itemize}
\end{theorem}

We will employ this theorem to prove the correctness of our abstraction of EHSs in the following section.

\section{Polynomial Abstraction of EHSs}\label{sec:abstraction}
In this section, given any EHS as defined in Definition \ref{dfn:ele-hs}, we will construct a PHS that simulates the EHS in the sense of Definition \ref{dfn:simu-HS}. The process of constructing such an abstraction can be divided into three steps: firstly, elementary ODEs can be transformed into polynomial forms by introducing new variables to replace non-polynomial terms occurring in the vector field functions; secondly, using the replacement relations, initial sets and domains, and thus EDSs, can be abstracted into polynomial forms; finally, discrete transitions, i.e. guards and reset functions, can be abstracted accordingly, which results in polynomial abstractions of EHSs.

\subsection{Polynomialization of Elementary ODEs}\label{subsec:poly-ode}
In this part, we illustrate how to transform an elementary ODE $\dot \xx = \fb(\xx)$ equivalently into a polynomial one. The basic idea is to introduce a fresh variable $v$ for each non-polynomial term $\gamma(\xx)$ in $\fb(\xx)$ and then substitute $v$ for $\gamma(\xx)$ in $\fb(\xx)$; meanwhile differentiate the two sides of the replacement equation $v=\gamma(\xx)$ w.r.t. time and obtain a new ODE $\dot{v}=\nabla\gamma(\xx)\cdot \fb(\xx)$, where $\nabla\gamma(\xx)$ denotes the \emph{gradient} row vector of $\gamma(\xx)$; then append the new ODE to the original one (with $\gamma(\xx)$ replaced by $v$), and continue the above procedure to replace non-polynomial terms that may exist in $\nabla\gamma(\xx)$; finally when such a process terminates, a polynomial ODE together with a collection of replacement equations will be obtained. Note that the transformed polynomial ODE will always have a \emph{higher} dimension than the original one.

\begin{remark}
 Recasting elementary ODEs as polynomial ones has been proposed in early works {in the field of physics and biosciences} such as \cite{Kerner81,SV87} {in order to obtain explicit solutions of EDSs. In this paper, we employ such an idea for formal verification and invariant generation for EHSs.} The basic transformation here is similar to \cite{Kerner81}, but we give a clearer statement of the transformation procedure and extend it from ODEs to hybrid systems.
\end{remark}

We next demonstrate the above idea on concrete examples.

\subsubsection{Univariate Basic Elementary Functions}
For
\begin{equation}\label{eqn:ode-poly-simple}
  \dot x = f(x)
\end{equation}
\begin{itemize}
  \item if $f(x)=\frac{1}{x}$, then let $v=\frac{1}{x}$, and thus $\dot v = - \frac{\dot x}{x^2}$. Therefore (\ref{eqn:ode-poly-simple}) is transformed to
        \begin{displaymath}
            \left\{ \begin{array}{lll}
                \dot x & = & v\\
                \dot v & = & -v^3\\
         \end{array} \right.\enspace ;\footnote{By $v=\frac{1}{x}$, the set $\{(x,v)\mid x=0, v\in\mathbb R\}$ is excluded from the domain of the transformed polynomial ODE. Such consequence will not be explicitly mentioned in the rest of this paper.}
        \end{displaymath}

  \item if $f(x)=\sqrt{x}$, then let $v=\sqrt{x}$, and thus $\dot v = \frac{\dot x}{2\sqrt{x}}$. Therefore (\ref{eqn:ode-poly-simple}) is transformed to
        \begin{displaymath}
            \left\{ \begin{array}{lll}
                \dot x & = & v\\
                \dot v & = & \frac{1}{2}\\
         \end{array} \right.\enspace ;
        \end{displaymath}

  \item if $f(x)=e^{x}$, then let $v=e^x$, and thus $\dot v = e^x\cdot \dot x$. Therefore (\ref{eqn:ode-poly-simple}) is transformed to
          \begin{displaymath}
            \left\{ \begin{array}{lll}
                \dot x & = & v\\
                \dot v & = & v^2\\
         \end{array} \right.\enspace ;
        \end{displaymath}

  \item if $f(x)=\ln{x}$, then let $v=\ln{x}$, and thus $\dot v = \frac{\dot x}{x}$; then further let $u=\frac{1}{x}$, and thus $\dot u = - \frac{\dot x}{x^2}$. Therefore (\ref{eqn:ode-poly-simple}) is transformed to
      \begin{displaymath}
            \left\{ \begin{array}{lll}
                \dot x & = & v\\
                \dot v & = & uv\\
                \dot u & = & -u^2 v
         \end{array} \right.\enspace ;
        \end{displaymath}

  \item if $f(x)=\sin{x}$, then let $v=\sin{x}$, and thus $\dot v = \dot x\cdot \cos{x}$; then further let $u=\cos{x}$, and thus $\dot u = - \sin{x}\cdot \dot x$. Therefore (\ref{eqn:ode-poly-simple}) is transformed to
      \begin{displaymath}
            \left\{ \begin{array}{lll}
                \dot x & = & v\\
                \dot v & = & uv\\
                \dot u & = & -v^2
         \end{array} \right.\enspace ;
        \end{displaymath}

    \item if $f(x)=\cos{x}$, then the transformation is analogous to the case of $f(x)=\sin{x}$.
\end{itemize}

\subsubsection{Compositional and Multivariate Functions}
Obviously, the outmost form of any compositional elementary function must be one of $f\pm g, f \times g, \frac{f}{g}, f^{a}, e^f, \ln(f), \sin(f), \cos(f)$. Therefore given a compositional function, we can iterate the above procedure discussed on basic cases from the innermost non-polynomial sub-term to the outside, until all the sub-expressions have been transformed into polynomials. For example,
\begin{itemize}
  \item if $f(x)=\ln({2+\sin{x}})$, we can let
       \begin{displaymath}
            \left\{ \begin{array}{ll}
                v &  =\sin{x}\\
                u & = \cos{x}\\
                w &  =\ln{(2+v)}=\ln{(2+\sin{x})}\\
                z & =\frac{1}{2+v}=\frac{1}{2+\sin{x}}\\
         \end{array} \right.\enspace ,
        \end{displaymath}
\end{itemize}
and then (\ref{eqn:ode-poly-simple}) is transformed to
        \begin{displaymath}
            \left\{ \begin{array}{lll}
                \dot x & = & w\\
                \dot v & = & uw\\
                \dot u & = & -vw\\
                \dot w & = & zuw\\
                \dot z & = & -z^2uw
         \end{array} \right.\enspace .
        \end{displaymath}
Handling multivariate functions is straightforward.

In summary, we give the following assertion on polynomializing elementary ODES,
the correctness of which can be given based on the formal transformation algorithms
presented in the appendix.
\begin{proposition}[Polynomial Recasting]\label{prop:ode-polynomialize}
  Given an ODE $\dot \xx=\fb(\xx)$ with $\fb(\xx)$ an elementary vector function defined on an open set $U\subseteq\RR^n$, there exists a
  collection of variable replacement equations $\vv=\Gamma(\xx)$, where
  $\vv=(v_1,v_2,\ldots, v_m)$ is a vector of new variables and $\Gamma(\xx)=(\gamma_1(\xx),\gamma_2(\xx),\ldots,\gamma_m(\xx)):U\rightarrow \RR^m$ is an elementary vector function, such that
  \begin{eqnarray}
    {\left( \begin{array}{l}
                \dot \xx \\
                \dot \vv
         \end{array} \right) }  =  \left( \begin{array}{c}
                \fb(\xx) \\
                \mathcal J_{\Gamma}(\xx)\cdot  \fb(\xx)
         \end{array} \right)  & = &  \left( \begin{array}{c}
                \fb(\xx) \\
                \mathcal J_{\Gamma}(\xx)\cdot  \fb(\xx)
         \end{array} \right) {\Big\llbracket \vv/\Gamma(\xx)\Big\rrbracket} \nonumber \\
        & \widehat{=} & \tilde{\fb}(\xx,\vv) \label{eqn:ode-poly-define}
  \end{eqnarray}
  becomes a polynomial ODE, that is, $\tilde \fb(\xx,\vv)$ is a polynomial vector function in variables $\xx$ and $\vv$. Here
  {\sf {expr}}$\llbracket \vv/\Gamma(\xx) \rrbracket$ means replacing any occurrence of the non-polynomial term $\gamma_i(\xx)$ in the expression {\sf {expr}} by the corresponding variable $v_i$, for all $1\leq i\leq m$.
\end{proposition}

It can be proved that the number of variables $\vv$ is at most triple the number of nonpolynomial
terms in the original ODE, which can be a small number in practice.
The transformed polynomial ODE as specified in Proposition \ref{prop:ode-polynomialize} is \emph{equivalent} to the original one in the following sense.

\begin{theorem}[Trajectory Equivalence]\label{thm:ode-equiv}
 Let $\fb(\xx)$, $\Gamma(\xx)$ and $\tilde \fb (\xx,\vv)$ be as specified in Proposition \ref{prop:ode-polynomialize}. Then for any trajectory $\xx(t)$ of $\dot\xx=\fb(\xx)$ starting from $\xx_0\in U\subseteq \RR^n$, $\big(\xx(t),\Gamma(\xx(t))\big)$ is the trajectory of $(\dot \xx,\dot \vv)=\tilde \fb(\xx,\vv)$ starting from $(\xx_0,\Gamma(\xx_0))$; conversely, for any trajectory $(\xx(t),\vv(t))$ of $(\dot \xx,\dot \vv)=\tilde \fb(\xx,\vv)$ starting from $(\xx_0,\vv_0)\in \RR^{n+m}$, if $\xx_0\in U$ and $\vv_0=\Gamma(\xx_0)$, then $\xx(t)$ is the trajectory of $\dot \xx=\fb(\xx)$ starting from $\xx_0$.
\end{theorem}
\begin{proof}
  The result can be deduced directly from (\ref{eqn:ode-poly-define}). \qed
\end{proof}

\subsection{Abstracting EDSs by PDSs}\label{subsec:poly-eds}
In this part, given an EDS $\mathcal C_{\xx}\,\define\,(\Xi_{\xx},\fb_{\xx},D_{\xx})$ we will construct a PDS
$C_{\yy}\,\define\,(\Xi_{\yy},\fb_{\yy},D_{\yy})$ that simulates $\mathcal C_{\xx}$. The construction is based on the procedure introduced in Section \ref{subsec:poly-ode} on polynomial transformation of elementary ODEs. The basic idea is to construct a simulation map using the replacement equations.
The difference here is that when abstracting an EDS, we need to replace non-polynomial terms occurring in not only the vector field, but also the initial set and domain. Roughly, the construction of $\mathcal C_{\yy}$ consists of the following four steps.

\begin{enumerate}
  \item[(S1)] Introduce new variables to replace all non-polynomial terms in $\fb_{\xx}$, $\Xi_{\xx}$ and $D_{\xx}$, and obtain a collection of replacement equations $\vv=\Gamma(\xx)$ such that $\fb_{\xx}\llbracket \vv/\Gamma(\xx)\rrbracket$,  $\Xi_{\xx}\llbracket \vv/\Gamma(\xx)\rrbracket$ and  $D_{\xx}\llbracket \vv/\Gamma(\xx)\rrbracket$ all become polynomial expressions.
  \item[(S2)] Differentiate both sides of $\vv=\Gamma(\xx)$ w.r.t. time to get $\dot \vv=J_{\Gamma}(\xx)\cdot  \fb(\xx)$, and replace all newly appearing non-polynomial terms by introducing more variables.
  \item[(S3)] Repeat (S2) until no more variables need to be introduced. For simplicity, still denote the final set of replacement equations by $\vv=\Gamma(\xx)$. By Proposition \ref{prop:ode-polynomialize}, a polynomial vector field $\tilde \fb(\xx,\vv)$ as in (\ref{eqn:ode-poly-define}) will be obtained. Let $\yy\,\define\,(\xx,\vv)$ and define
      \begin{equation}\label{eqn:abs-ode}
        \fb_{\yy}(\xx,\vv)\quad\define\quad \tilde\fb(\xx,\vv)\enspace .
      \end{equation}
  \item[(S4)] Define the simulation map $\Theta:D_{\xx}\rightarrow \RR^{|\yy|}$ as\footnote{Here we assume that all elementary functions in $\Xi_{\xx},\fb_{\xx}$ and $D_{\xx}$ are defined on $D_{\xx}$.}
  \begin{equation}\label{eqn:abs-theta}
    \Theta (\xx) = (\xx, \Gamma(\xx))\enspace.
  \end{equation}
  Then use $\Theta$ to construct $\Xi_{\yy}$ and $D_{\yy}$ as illustrated later.
\end{enumerate}

After the above four steps, a CDS $(\Xi_{\yy},\fb_{\yy},D_{\yy})$ will be obtained, which is intended to be the polynomial abstraction of $\mathcal C_{\xx}$ under simulation map $\Theta$. We next show how to get $\Xi_{\yy}$ and $D_{\yy}$ in detail to complete the construction.

The image of $\Xi_{\xx}$ under the simulation map $\Theta$ is
$$\Theta(\Xi_{\xx})=\{(\xx,\vv)\in\RR^{|\yy|}\mid \xx\in\Xi_{\xx}\wedge \vv=\Gamma(\xx)\}\,,$$
briefly denoted by
$\Theta(\Xi_{\xx})\,\,\,\define\,\,\,\Xi_{\xx}\wedge \vv=\Gamma(\xx)$, or alternatively
\begin{equation}\label{eqn:abs-xi-1}
  \Theta(\Xi_{\xx})\,\,\,\define\,\,\,\Xi_{\xx}\llbracket\vv/\Gamma(\xx)\rrbracket \wedge \vv=\Gamma(\xx)\,.
\end{equation}
By (S1), the first conjunct in (\ref{eqn:abs-xi-1}) is of polynomial form, but the second conjunct contains elementary functions. By Definition \ref{dfn:simu-cds}, we need to get a polynomial over-approximation $\Xi_{\yy}$ of $\Theta(\Xi_{\xx})$, which means we need to abstract $\vv=\Gamma(\xx)$ in (\ref{eqn:abs-xi-1}) by polynomial expressions.
We propose the following four ways to do so.
\begin{itemize}
  \item[(W1)]  When $\Gamma(\xx)$ are some special kinds elmentary functions, $\vv= \Gamma(\xx)$ can be equivalently transformed to polynomial expressions, e.g.          \begin{equation}\label{eqn:abs-gamma-1}
      \left\{ \begin{aligned}
               & v= \frac{1}{x} \Longleftrightarrow vx=1 \\
                & v= \sqrt{x} \Longleftrightarrow v^2=x \,\wedge\, v\geq 0 \\
               \end{aligned}
      \right. \enspace .
      \end{equation}
  \item[(W2)] If $\Xi_{\xx}$ is a bounded region and the upper/lower bounds of each component $x_i$ of $\xx$ can be easily obtained, then we can compute the Taylor polynomial expansion $\mathbf p(\xx)$ of $\Gamma(\xx)$ over the bounded region up to a certain degree, as well as an interval over-approximation $\mathbf I$ of the corresponding truncation error, such that $\vv=\Gamma(\xx)$ can be approximated by $\vv\in (\mathbf p(\xx) + \mathbf I)$. We will illustrate this by an example presented later.
  \item[(W3)] We can also just compute the range of $\Gamma(\xx)$ (over $\Xi_{\xx}$) as an over-approximation of $\vv$, e.g.
            \begin{equation}\label{eqn:abs-gamma-2}
      \left\{ \begin{aligned}
               & v= \sin{x} \Longrightarrow  -1\leq v \leq 1 \\
                & v= e^{x} \Longrightarrow v>0 \\
               \end{aligned}
      \right. \enspace .
      \end{equation}
  \item[(W4)] The simplest way is to remove the constraint $\vv=\Gamma(\xx)$ entirely, which means $\vv$ is allowed to take any value from $\RR^{|\vv|}$.
\end{itemize}

From (W1) to (W4), the over-approximation of $\vv=\Gamma(\xx)$ becomes more and more coarse. Usually it takes more effort to obtain a more refined abstraction, but the result would be more helpful for analysis of the original system. We will discuss in Section \ref{sec:application} how to choose among (W1)-(W4) when constructing abstractions of EHSs, depending on what kind of inductive invariants are to be generated for safety verification tasks.

The construction of $D_{\yy}$ is the same as $\Xi_{\yy}$. Then we can give the following conclusion.
\begin{theorem}[Abstracting EDS by PDS]\label{thm:abs-eds-pds}
  Given an EDS $\mathcal C_{\xx}\,\define\,(\Xi_{\xx},\fb_{\xx},D_{\xx})$, let $\mathcal C_{\yy}\,\define\,(\Xi_{\yy},$ $\fb_{\yy},D_{\yy})$, where $\fb_{\yy}$ is given by (\ref{eqn:abs-ode}) and (\ref{eqn:ode-poly-define}), and $\Xi_{\yy},D_{\yy}$ are given by (8) together with (W1)-(W4). Then $\mathcal C_{\yy}$ is a \emph{polynomial} abstraction of $\mathcal C_{\xx}$ in the sense of Definition \ref{dfn:simu-cds}, under simulation map $\Theta$ defined by (\ref{eqn:abs-theta}).
\end{theorem}
\begin{proof}
  First, it is easy to check that $\mathcal C_{\yy}$ is a PDS and $\Theta(\Xi_{\xx})\subseteq \Xi_{\yy}$, $\Theta(D_{\xx})\subseteq D_{\yy}$.
Second, by (\ref{eqn:abs-theta}) we have
  $$\mathcal J_{\Theta}(\xx)=
     \left( \begin{array}{c}
                \mathbf{Id}_{|\xx|} \\
                \mathcal J_{\Gamma}(\xx)
         \end{array} \right)\enspace ,
  $$
  where $\mathbf{Id}_{|\xx|}$ denotes the $|\xx|$-dimensional identity matrix. Then for any $\xx\in D_{\xx}$,
  $$\mathcal J_{\Theta}(\xx) \cdot \fb_{\xx}(\xx)=
       \left( \begin{array}{c}
                \mathbf{Id}_{|\xx|} \\
                \mathcal J_{ \Gamma}(\xx)
         \end{array} \right) \cdot \fb_{\xx}(\xx)
                =\left( \begin{array}{c}
                \fb_{\xx}(\xx) \\
                \mathcal J_{\Gamma}(\xx)\cdot  \fb_{\xx}(\xx)
         \end{array} \right)
         \enspace .
  $$
  Then according to the above formula and (\ref{eqn:abs-theta}), (\ref{eqn:abs-ode}) and  (\ref{eqn:ode-poly-define}), we get
$\fb_{\yy}(\Theta(\xx))=\fb_{\yy}(\xx,\Gamma(\xx))=\tilde\fb(\xx,\Gamma(\xx))=\mathcal J_{\Theta}(\xx) \cdot \fb_{\xx}(\xx)$. Therefore by Theorem \ref{thm:sim-check} we get the conclusion. \qed
\end{proof}

\begin{example}\label{eg:eds-pds}
Consider the EDS $\mathcal C_{\xx}\,\define\,(\Xi_{\xx},\fb_{\xx}, D_{\xx})$, where
\begin{itemize}
  \item[--]  $\Xi_{\xx}\,\define\,(x+0.5)^2+(y-0.5)^2-0.16\leq 0$;
  \item[--]  $D_{\xx}\,\define\,-2\leq x\leq 2 \wedge -2 \leq y \leq 2$; and
  \item[--]  $\fb_{\xx}$ defines the ODE
\begin{equation}\label{eqn:eg-transode}
\left(\begin{array}{c}
\dot x\\
\dot y
\end{array}
\right) = \left( \begin{array}{c}
e^{-x} + y - 1 \\
-\sin^2(x)
\end{array} \right)\enspace.
\end{equation}
\end{itemize}
\end{example}

We will show how to construct a PDS $\mathcal C_{\yy}$ that simulates $\mathcal C_{\xx}$ following the above described steps.
\begin{itemize}
  \item (S1-S3): Noticing that $\Xi_{\xx}$ and $D_{\xx}$ are both in polynomial forms, we only need to replace non-polynomial terms in $\fb_{\xx}$. We finally obtain the replacement relations $\vv=\Gamma(\xx)$ given by
      \begin{equation}\label{eqn:eg-eds-gamma}
      (v_1,v_2,v_3)= (\sin{x},e^{-x},\cos x)\end{equation}
      and the transformed polynomial ODE
      \begin{equation}\label{eqn:eg-transode-poly}
\left(\begin{array}{c}
\dot x\\
\dot y\\
\dot v_1\\
\dot v_2\\
\dot v_3
\end{array}
\right)
= \left( \begin{array}{c}
v_2 + y - 1 \\
-v_1^2\\
v_3(v_2+y-1)\\
-v_2(v_2+y-1)\\
-v_1(v_2+y-1)
\end{array} \right)\enspace ,
\end{equation}
the right-hand-side of which is defined to be $\fb_{\yy}$.
\item (S4): The simulation map $\Theta$ is given by
$$\Theta(x,y)=(x,y,\sin{x},e^{-x},\cos x)\enspace .$$
The images of $\Xi_{\xx}$ and $D_{\xx}$ under $\Theta$ are
$$
\Theta(\Xi_{\xx})\,\,\define\,\,\Xi_{\xx}\wedge v_1=\sin{x}\wedge v_2=e^{-x}\wedge v_3=\cos{x}
$$
and
$$\Theta(D_{\xx})\,\,\define\,\,D_{\xx}\wedge v_1=\sin{x}\wedge v_2=e^{-x}\wedge v_3=\cos{x}$$
respectively. For the above two formulas, (W1) is not applicable, whereas we can use any of (W2)-(W4) to abstract them. Here we just give one possible way. First, use (W4) to abstract $\Theta(\Xi_{\xx})$, and define
$\Xi_{\yy}\,\,\define\,\,\Xi_{\xx}$. Next, adopt (W2) to abstract $\Theta(D_{\xx})$; using the tool {\small \sf COSY INFINITY}\footnote{\url{http://bt.pa.msu.edu/index_cosy.htm}} for \emph{Taylor model} \cite{KM03} computation, we expand $\sin{x}$, $e^{-x}$ and $\cos{x}$ over $x\in [-2,2]$ at point $x=0$ up to degree 6,
and obtain
\begin{eqnarray}
  & & p_1(x)+l_1\leq v_1\leq p_1(x)+u_1 \label{eqn:taylor-sin-x} \\
 \mathit{TM}_{\xx,\vv}\, \define  & \wedge &  p_2(x)+l_2\leq v_2\leq p_2(x)+u_2 \label{eqn:taylor-exp-x}\enspace \enspace \enspace \enspace \\
   & \wedge &  p_3(x)+l_3\leq v_3\leq p_3(x)+u_3 \enspace .\nonumber
\end{eqnarray}
Figure \ref{fig:taylor-model} is an illustration of the relations between $v_1, v_2$ and $x$ given by (\ref{eqn:taylor-sin-x}) and (\ref{eqn:taylor-exp-x}) respectively,
where
\begin{itemize}
  \item $p_1(x)\,=\,2 (0.5 x) - 1.333333333333333  (0.5 x)^3 + 0.2666666666666667 (0.5 x)^5$
  \item $l_1=-0.08888888888890931$
  \item $u_1=0.08888888888890931$
\end{itemize}
and
\begin{itemize}
  \item  $p_2(x)\,= \,1-2(0.5x)+2(0.5x)^2 \,- 1.333333333333333(0.5x)^3$
  \item[] \qquad\qquad$+0.6666666666666666(0.5x)^4-0.2666666666666667(0.5x)^5$
  \item[] \qquad\qquad$+0.08888888888888889(0.5x)^6$
  \item  $l_2=-0.1876585675919477$
  \item  $u_2=0.1876585675919477$ \enspace .
\end{itemize}

\begin{figure}
\begin{center}
\includegraphics[width=1.6in,height=1.5in]{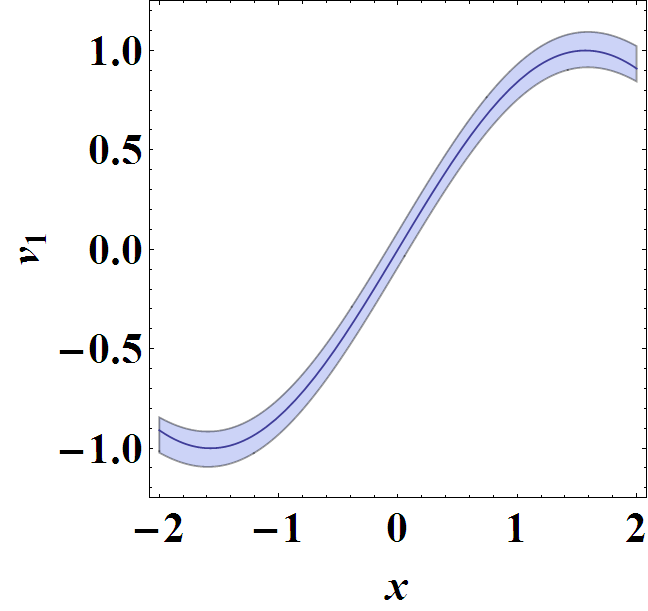}
\hspace{1cm}
\includegraphics[width=1.5in,height=1.52in]{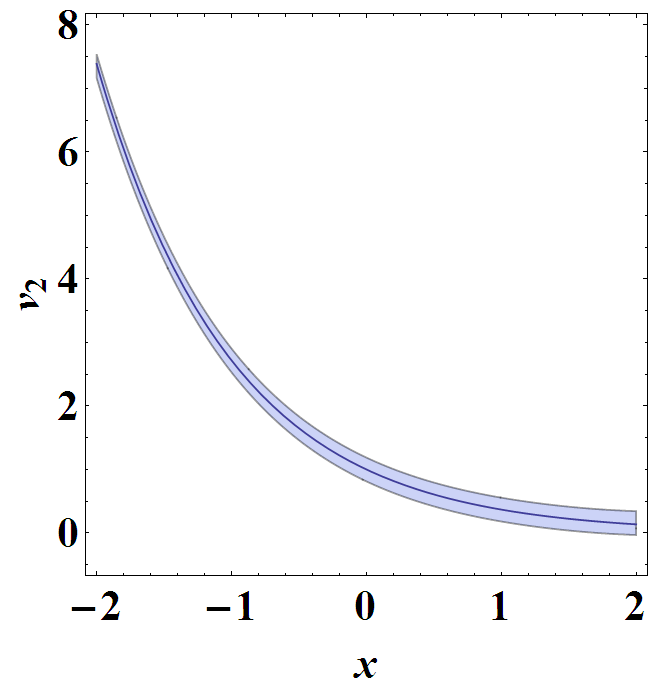}
\caption{Taylor polynomial approximation of elementary functions}
\label{fig:taylor-model}
\end{center}
\end{figure}

Then we can define
$D_{\yy}\,\,\define\,\, D_{\xx}\wedge \mathit{TM}_{\xx,\vv}$.
Thus we finally get a PDS $\mathcal C_{\yy}\,\define\,(\Xi_{\yy},\fb_{\yy},$ $D_{\yy})$ that simulates $\mathcal C_{\xx}$. Note that here $\Xi_{\yy}$ is not a subset of $D_{\yy}$, which conflicts with our assumption on CDSs. However, allowing more behavior in $\mathcal C_{\yy}$ does not affect the soundness of abstraction; besides, this problem can be easily remedied by taking $\Xi_{\yy}\wedge D_{\yy}$ as the initial set. We keep the current form for ease of safety verification in Section \ref{sec:application}.
\end{itemize}

\subsection{Abstracting EHSs by PHSs}\label{subsec:poly-ehs}

In the previous sections, we have presented a method to abstract an EDS to
a PDS such that the PDS simulates the EDS.
Now we show, given an EHS $\mathcal H_{\xx}$, how to construct a simulation map $\Theta$ and the corresponding PHS $\mathcal H_{\yy}$ that simulates $\mathcal H_{\xx}$.
Actually, this can be easily done by just extending the previous abstraction approach a bit to take into account guard constraints and reset functions.
Another difference is that we need to treat each mode of a HA separately by constructing an individual simulation map for each of them.

More specifically, given an EHS $\mathcal{H}_{\xx}\,\define\,(Q,X,f_{\xx},D_{\xx},$
$E,G_{\xx},R_{\xx},$ $\Xi_{\xx})$, for each mode $q\in Q$, and for any $e\in E$ with $q$ the starting mode, we need to introduce new variables to replace all non-polynomial terms occurring in $\fb_{\xx,q}$, $\Xi_{\xx,q}$, $D_{\xx,q}$, $G_{\xx,e}$ and $R_{\xx,e}$, and then compute the time derivatives of the fresh variables, as we did in the continuous case. In this way, for each mode $q$ we will obtain a vector of new variables $\vv_q$ and the corresponding replacement equations $\vv_q=\Gamma_{q}(\xx)$; without loss of generality, we can assume all $\vv_q$ to be of the same dimension, and thus can get rid of the subscript $q$ of $\vv_q$. At the same time, for all mode $q$, the elementary vector field $\fb_{\xx,q}$ will be transformed into a polynomial one, i.e. $\tilde \fb_q(\xx,\vv)$, as given by Proposition \ref{prop:ode-polynomialize} and formula (\ref{eqn:ode-poly-define}).

Let $\yy\,\define\,(\xx,\vv)$. Let $\Theta_q: D_{\xx,q}\rightarrow \RR^{|\yy|}$ be given by\footnote{Here we assume that for all $q\in Q$ and $e=(q,q')\in E$, the elementary functions in  $\fb_{\xx,q}$, $\Xi_{\xx,q}$, $D_{\xx,q}$, $G_{\xx,e}$ and $R_{\xx,e}$ are well defined on $D_{\xx,q}$.}
\begin{equation}\label{eqn:abs-hs-theta}
  \Theta_q (\xx) = (\xx,\Gamma_q(\xx))\enspace .
\end{equation}
Now the construction of $\mathcal{H}_{\yy}\,\define\,(Q,Y,f_{\yy},D_{\yy}, E,G_{\yy},R_{\yy},\Xi_{\yy})$ can proceed as follows.
\begin{itemize}
  \item For each $q\in Q$, let
    \begin{equation}\label{eqn:abs-hs-fb}
      \fb_{\yy,q}(\xx,\vv)\,\,\define\,\,\tilde \fb_q(\xx,\vv)
    \end{equation}
    with $\tilde \fb_q(\xx,\vv)$ given by (\ref{eqn:ode-poly-define}).
  \item For each $q\in Q$, abstract $\vv=\Gamma_q(\xx)$ in
    \begin{equation}\label{eqn:abs-hs-xi}
     \Xi_{\xx,q}\llbracket\vv/\Gamma_q(\xx)\rrbracket \wedge \vv=\Gamma_q(\xx)
    \end{equation}
    and
    \begin{equation}\label{eqn:abs-hs-domain}
     D_{\xx,q}\llbracket\vv/\Gamma_q(\xx)\rrbracket \wedge \vv=\Gamma_q(\xx)
    \end{equation}
    by polynomial expressions along the ways (W1)-(W4),
    and thus $\Xi_{\yy,q}$ and $D_{\yy,q}$ can be obtained.
  \item For each $e\in E$ with $q$ the starting mode, abstract $\vv=\Gamma_q(\xx)$ in
    \begin{equation}\label{eqn:abs-hs-guard}
     G_{\xx,e}\llbracket\vv/\Gamma_q(\xx)\rrbracket \wedge \vv=\Gamma_q(\xx)
    \end{equation}
    by polynomial expressions along the ways (W1)-(W4),
    and thus $G_{\yy,e}$ can be obtained.
\end{itemize}

So far, the only component left unspecified in $\mathcal H_{\yy}$ is $R_{\yy,e}$.

\begin{itemize}
  \item For each $e=(q,q')\in E$, define
\begin{eqnarray}
\tilde {R}_{\yy,e}(\xx,\vv)
\!\!\!& \,\define & \!\!\! \{ (\xx',\vv')\mid \xx'\in R_{\xx,e}(\xx)\llbracket \vv/\Gamma_q(\xx) \rrbracket \wedge\, \enspace \nonumber \\
& &\!\quad \qquad \quad \vv' = \Gamma_{q'}(\xx') \,\} \enspace . \label{eqn:abs-hs-reset}
\end{eqnarray}
Then abstract $\vv' = \Gamma_{q'}(\xx')$ in (\ref{eqn:abs-hs-reset}) by polynomial expressions along the ways (W1)-(W4),
and thus $R_{\yy,e}$ can be obtained. For example, if (W4) is adopted then $R_{\yy,e}$ can be defined as
$$R_{\yy,e}(\xx,\vv)\,\,\,\define\,\,\,\{(\xx',\vv')\mid \xx'\in R_{\xx,e}(\xx)\llbracket \vv/\Gamma_q(\xx) \rrbracket\}\enspace .$$
In particular, if $R_{\xx,e}$ is an identity map and $\Gamma_{q}=\Gamma_{q'}$,
then $R_{\yy,e}$ is also an identity map.
\end{itemize}

\begin{theorem}[Abstracting EHS by PHS]\label{thm:abs-ehs-phs}
Given an EHS $\mathcal{H}_{\xx}\,\define\,(Q,X,f_{\xx},D_{\xx}, E,G_{\xx},$ $R_{\xx},\Xi_{\xx})$, let $\mathcal{H}_{\yy}\, \define\, (Q,Y,$ $f_{\yy}, D_{\yy},$ $E,$ $G_{\yy}, R_{\yy}, \Xi_{\yy})$, where $\fb_{\yy,q}$ is given by (\ref{eqn:abs-hs-fb}) and (\ref{eqn:ode-poly-define}), and $\Xi_{\yy,q},D_{\yy,q},$ $G_{\yy,e},R_{\yy,e}$ are given by (\ref{eqn:abs-hs-xi}), (\ref{eqn:abs-hs-domain}), (\ref{eqn:abs-hs-guard}) and (\ref{eqn:abs-hs-reset}), together with (W1)-(W4), respectively. Then $\mathcal H_{\yy}$ is a \emph{polynomial} abstraction of $\mathcal H_{\xx}$ in the sense of Definition \ref{dfn:simu-HS}, under the simulation maps $\Theta_q$ defined by (\ref{eqn:abs-hs-theta}).
\end{theorem}
\begin{proof}
 First, it is easy to check that $\mathcal H_{\yy}$ is a PHS. Second,
 by Theorem \ref{thm:abs-eds-pds}, we can get
 $(\Xi_{{\yy},q},\fb_{{\yy},q}, D_{{\yy},q})$ simulates $(\Xi_{{\xx},q}, \fb_{{\xx},q},$ $D_{{\xx},q})$ via $\Theta_q$, for each $q\in Q$.
 Third, from (\ref{eqn:abs-hs-guard}) and (W1)-(W4) it is easy to see that
 $\Theta_q( G_{{\xx},e})\subseteq  G_{{\yy},e}$, for any $e=(q,q')\in E$.
 By Definition \ref{dfn:simu-HS}, we finally need to show that
$\Theta_{q'}( R_{{\xx},e}({\xx}))\subseteq R_{{\yy},e}(\Theta_q({\xx}))$, for any $e=(q,q')\in E$ and any ${\xx}\in G_{\xx,e}$.

By (\ref{eqn:abs-hs-theta}) we have
\begin{equation}\label{eqn:proof-ehs-pds}
\Theta_{q'}( R_{{\xx},e}({\xx}))=\{(\xx',\vv')\mid \xx'\in R_{{\xx},e}({\xx}) \wedge \vv'=\Gamma_{q'}(\xx')\}\enspace.
\end{equation}
By (\ref{eqn:abs-hs-reset}), (\ref{eqn:abs-hs-theta}) and (\ref{eqn:proof-ehs-pds}) we have
$$\tilde R_{\yy,e}(\Theta_q(\xx))=\tilde R_{\yy,e}(\xx,\Gamma_q(\xx))=\Theta_{q'}( R_{{\xx},e}({\xx}))\enspace .$$
By (W1)-(W4) we have $\tilde R_{\yy,e}(\Theta_q(\xx))\subseteq R_{\yy,e}(\Theta_q(\xx))$. Therefore we finally get \\$\Theta_{q'}( R_{{\xx},e}({\xx}))\subseteq R_{{\yy},e}(\Theta_q({\xx}))$. \qed
\end{proof}

\begin{example}
  Consider the example of a bouncing ball over a sine-waved surface as illustrated by the left picture in Figure \ref{fig:bouncing}, adapted from a similar one in \cite{Hydlogic}. The motion of the ball stays in the two-dimensional $x$-$y$ plane, with $x$ denoting the horizontal position and $y$ denoting the height, and the velocity along the two directions are denoted by $v_x$ and $v_y$ respectively. When the ball hits the surface given by the sine wave $y=\sin{x}$, its dynamics changes instantaneously. We assume the collision between the ball and the surface to be perfectly elastic so that there is no loss of energy. For instance, if the ball touches the surface at point $(0,0)$ with a downward vertical velocity $v_y$ and zero horizontal velocity $v_x$, then after collision $v_y$ becomes $0$ while $v_x$ takes the value of $v_y$ before collision.

\begin{figure}
\begin{center}
\includegraphics[width=1.6in,height=1.5in]{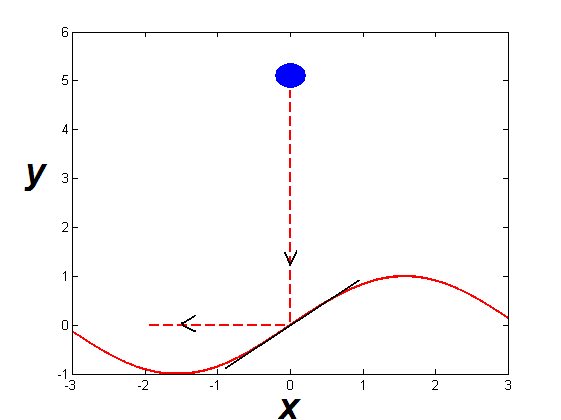}
\hspace{1cm}
\includegraphics[width=1.6in,height=1.5in]{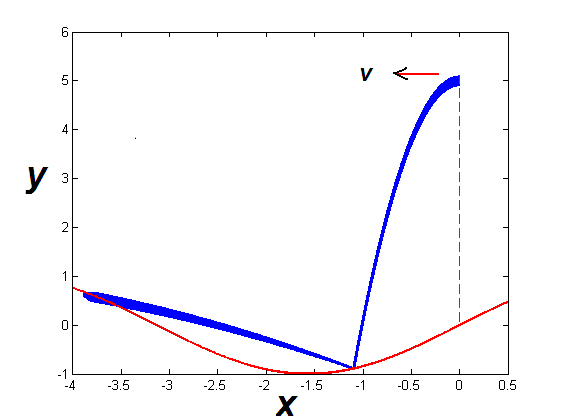}
\caption{Bouncing ball on a sine-waved surface}
\label{fig:bouncing}
\end{center}
\end{figure}

  As explained above, the HA model $\mathcal H_{\xx}$ of the bouncing ball can be given as
  \begin{itemize}
    \item $Q=\{q\}$; $X=\{x,y,v_x,v_y\}$;
    \item $E=\{e\}$ with $e = (q,q)$;
    \item $D_{\xx,q}\,\define\,y\geq \sin{x}$; $G_{\xx,e}\,\define\,y=\sin{x}$;
    \item $\Xi_{\xx,q}\,\define\,y\geq 4.9\wedge y\leq 5.1\wedge x=0\wedge v_x=-1\wedge v_y=0$;
    \item $\fb_{\xx,q}$ defines the ODE
            \begin{equation}\label{eqn:bouncing-ode-1}
            \left\{ \begin{array}{lll}
                \dot x & = & v_x\\
                \dot y & = & v_y\\
                \dot v_x & = & 0\\
                \dot v_y & = & -9.8
         \end{array} \right.\enspace ;
        \end{equation}
    \item $R_{\xx,e}(x,y,v_x,v_y)\,\define\,\{(x,y,v_x',v_y')\}$ with
      \begin{equation}\label{eqn:bouncing-reset-1}
            \left\{ \begin{array}{lll}
                v_x' & = & \frac{(\sin{x})^2\cdot v_x + 2(\cos{x})\cdot v_y}{1+(\cos{x})^2}\\
                v_y' & = & \frac{2(\cos{x})\cdot v_x - (\sin{x})^2 \cdot v_y}{1+(\cos{x})^2}\\
         \end{array} \right.\enspace .
        \end{equation}
  \end{itemize}

Note that in the above model, non-polynomial expressions exist in $D_q,G_e$ and $R_e$. By applying our proposed abstraction approach, we obtained the replacement equations $(u_1,u_2,u_3)=(\sin{x},$ $\cos{x},\frac{1}{1+(\cos{x})^2})$, and the PHS $\mathcal H_{\yy}$:
\begin{itemize}
  \item $Q$ and $E$ are the same as $\mathcal H_{\xx}$;
  \item $Y=\{x,y,v_x,v_y,u_1,u_2,u_3\}$;
  \item $D_{\yy,q}\,\define\,y\geq u_1$; $G_{\yy,e}\,\define\,y=u_1$; note that here we adopt (W4) when abstracting $D_{\xx,q}$ and $G_{\xx,e}$;
  \item $\Xi_{\yy,q}\,\define\,\Xi_{\xx,q}\wedge u_1=0\wedge u_2=1\wedge u_3=0.5$;
  \item $\fb_{\yy,q}$ defines the ODE
            \begin{equation}\label{eqn:bouncing-ode-1}
            \left\{ \begin{array}{lll}
                \dot x & = & v_x\\
                \dot y & = & v_y\\
                \dot v_x & = & 0\\
                \dot v_y & = & -9.8\\
                \dot u_1 & = & u_2v_x\\
                \dot u_2 & = & -u_1v_x\\
                \dot u_3 & = &  2u_1u_2u_3^2v_x
         \end{array} \right.\enspace ;
        \end{equation}
  \item $R_{\yy,e}(\yy)\,\define\,\{(x,y,v_x',v_y',u_1,u_2,u_3)\}$ with
      \begin{equation}\label{eqn:bouncing-reset-1}
            \left\{ \begin{array}{lll}
                v_x' & = & u_3\cdot(u_1^2\cdot v_x + 2u_2\cdot v_y)\\
                v_y' & = & u_3\cdot (2u_2\cdot v_x - u_1^2\cdot v_y)\\
         \end{array} \right.\enspace;
        \end{equation}
         note that $u_1,u_2,u_3$ are only related to $x$ which is reset to itself, and thus the resets of $u_1,u_2,u_3$ are identity mappings.
\end{itemize}

Once we get the polynomial abstraction $\mathcal H_{\yy}$, we can use existing tools for PHSs to analyze its behavior. Here we use the state-of-the-art nonlinear hybrid system analyzer Flow$^*$ \cite{CAS13}. The right picture in Figure \ref{fig:bouncing} shows the computed reachable set over-approximation (projected to the $x$-$y$ plane) of $\mathcal H_{\yy}$ within two jumps, which is also the reachable set over-approximation of $\mathcal H_{\xx}$ by Theorem \ref{thm:SimuInvHa}. Note that such an analysis would NOT have been possible {directly on $\mathcal H_{\xx}$} in Flow$^*$ since its current version does not support elementary functions in domains, guards, or reset functions\footnote{Although Flow$^*$ does support nonlinear continuous dynamics with non-polynomial terms such as sine, cosine, square root, etc.}.
\end{example}

\section{Application in Safety Verification of EHSs}\label{sec:application}

One of the mostly studied problems in the study of HSs is safety verification. Given a HS $\mathcal H$, a safety requirement for $\mathcal H$ can be specified as $\mathcal S\,\define\,\bigcup_{q\in Q}(\{q\}\times S_q)$ with $S_q\subseteq \RR^n$ the safe region of mode $q$. Alternatively, a safety property can be
given as a set of unsafe regions $\mathcal {US}\,\define\,\bigcup_{q\in Q}(\{q\}\times \bar {S_q})$ with $\bar {S_q}$ the complement of $S_q$ in $\RR^n$. The safety verification problem asks whether $\mathcal R_{\mathcal H}\subseteq \mathcal S$, or equivalently, whether $\mathcal R_{\mathcal H}\cap \mathcal {US}=\emptyset$.

The following result relates the safety verification problem of a HS $\mathcal H_{\xx}$ to that of $\mathcal H_{\yy}$ which simulates $\mathcal H_{\xx}$.
\begin{theorem}[Safety Relation]\label{thm:safety-relation}
   Let $\mathcal {US}_{\xx}\,\define\,\bigcup_{q}(\{q\}\times \bar S_{\xx,q})$ be a safety requirement of the HS $\mathcal H_{\xx}$. Suppose $\mathcal H_{\yy}$ simulates $\mathcal H_{\xx}$ via simulation maps $\{\Theta_q\mid q\in Q\}$. Let $\mathcal {US}_{\yy}\,\define\,\bigcup_{q}(\{q\}\times \bar S_{\yy,q})$ with $\bar S_{\yy,q}\supseteq \Theta_q(\bar S_{\xx,q})$. Then if $\mathcal H_{\yy}$ is safe w.r.t. $\mathcal {US}_{\yy}$, then  $\mathcal H_{\xx}$ is safe w.r.t. $\mathcal {US}_{\xx}$.
\end{theorem}
\begin{proof}
  Let $\mathcal R_{\xx}$ and $\mathcal R_{\yy}\,\define\,\bigcup_{q}(\{q\}\times H_{\yy,q})$ denote the reachable sets of $\mathcal H_{\xx}$ and $\mathcal H_{\yy}$ respectively. Suppose $\mathcal R_{\yy}\cap \mathcal {US}_{\yy}=\emptyset$, i.e. $H_{\yy,q}\cap \bar S_{\yy,q}=\emptyset$ for any $q\in Q$.
  Thus $\Theta_q^{-1}(H_{\yy,q})$ $\cap\,\Theta_q^{-1}(\bar S_{\yy,q})=\emptyset$, which implies $$\Theta_q^{-1}(H_{\yy,q})\cap\,\Theta_q^{-1}\big(\Theta_q(\bar S_{\xx,q})\big)=\emptyset\enspace.$$
  Therefore $\Theta_q^{-1}(H_{\yy,q})\cap\bar S_{\xx,q}=\emptyset$.
  By Theorem \ref{thm:SimuInvHa} we get
  $\mathcal R_{\xx}\subseteq \bigcup_{q\in Q}\big(\{q\}\times \Theta_q^{-1}(H_{\yy,q})\big)$.
  Thus $\mathcal R_{\xx}\cap \mathcal {US}_{\xx}=\emptyset$.\qed
\end{proof}

Note that if the safety properties of EHSs are not in polynomial forms but contain elementary functions, we can replace the non-polynomial terms by new variables when constructing the simulation map, as we do for the EHSs themselves.

Theorem \ref{thm:safety-relation} allows us to take advantage of constraint-based approaches for PHSs to verify safety properties of EHSs.
In the rest of this section, we show how to perform safety verification for EHSs by combining the previous proposed polynomial abstraction method with constraint-based verification techniques for PHSs.

\subsection{Generating Polynomial Invariants}

In this and next subsections, for simplicity, we will use EDSs as special cases of EHSs to illustrate how to generate inductive invariants for safety verification of EHSs.

Given an EDS $\mathcal C_{\xx}\,\define\,(\Xi_{\xx},\fb_{\xx},D_{\xx})$ and an unsafe region $\bar S_{\xx}$, we first construct a PDS $\mathcal C_{\yy}\,\define\,(\Xi_{\yy},\fb_{\yy},D_{\yy})$ that simulates $\mathcal C_{\xx}$, as well as the polynomial abstraction $\bar S_{\yy}$ of $\bar S_{\xx}$.
According to Theorem \ref{thm:SimuInvHa} and \ref{thm:safety-relation},
if we can find a semi-algebraic\footnote{A set $A\subseteq \mathbb R^n$ is called {\emph{semi-algebraic}} if it can be defined by Boolean combinations of polynomial equations or inequalities.} inductive invariant $P(\yy)=P(\xx,\vv)$ for $\mathcal C_{\yy}$ with $\vv=\Gamma(\xx)$ the replacement equations, such that $P(\xx,\vv)$ is a certificate of the safety of $\mathcal C_{\yy}$ w.r.t. $\bar S_{\yy}$, then $P(\xx,\Gamma(\xx))$ is an inductive invariant certificate of the safety of $\mathcal C_{\xx}$ w.r.t. $\bar S_{\xx}$. If $P(\xx,\vv)$ does contain variables $\vv$, then $P(\xx,\Gamma(\xx))$ gives an elementary invariant of $\mathcal C_{\xx}$; otherwise $P(\xx,\Gamma(\xx))$ is just a polynomial invariant.

The form of the invariant $P(\yy)$ of $\mathcal C_{\yy}$ determines not only what kinds of invariants we can get for $\mathcal C_{\xx}$, but also the selection of abstraction ways (W1)-(W4) in Section \ref{subsec:poly-eds}.
To see this, we first assume for $\mathcal C_{\yy}$ a polynomial invariant candidate $P(\uu,\xx)\,\define\,p(\uu,\xx)\leq 0$ without the fresh variables $\vv$, where $\uu$ is the vector of parameters to be determined. Then a typical set of constraints on $\uu$ given by the constraint-based verification approach could be as follows:
\begin{itemize}
  \item[(C1)] $\forall \xx\forall \vv. (\Xi_{\yy}\longrightarrow P(\uu,\xx))$;
  \item[(C2)] $\forall \xx\forall\vv. (D_{\yy} \longrightarrow \nabla p(\uu,\xx) \cdot \fb_{\yy} \leq 0)$;
  \item[(C3)] $\forall \xx\forall \vv. (P(\uu,\xx) \longrightarrow \neg \bar S_{\yy})$.
\end{itemize}
By (W1)-(W4), it is easy to check that $(\exists \vv. \Xi_{\yy})\Longleftrightarrow \Xi_{\xx}$. Then we can prove that
(C1) is equivalent to $\forall \xx.(\Xi_{\xx}\longrightarrow P(\uu,\xx))$.
Similarly, by $(\exists \vv. \bar S_{\yy})\Longleftrightarrow \bar S_{\xx}$, we can prove that (C3) is equivalent to $\forall \xx.(P(\uu,\xx) \longrightarrow \neg \bar S_{\xx})$. Therefore we can conclude that it is sufficient to adopt (W4) for the abstraction of $\Xi_{\xx}$ and $\bar S_{\xx}$. The gradient $\nabla p(\uu,\xx)$ in (C2) is computed w.r.t. variables $\yy=(\xx,\vv)$. Since $p(\uu,\xx)$ does not contain $\vv$, all the partial derivatives of $p(\uu,\xx)$ w.r.t. $\vv$ are zero. The consequence of this fact is twofold: first, only those components of $\fb_{\yy}$ that define the derivatives of $\xx$, i.e. $\fb_{\xx}\llbracket\vv/ \Gamma(\xx) \rrbracket$, are relative to the computation of $p(\uu,\xx) \cdot \fb_{\yy}$, which means we do not even need to compute the derivatives of the fresh variables $\vv$ when constructing $\fb_{\yy}$; second, only those fresh variables occurring in $\fb_{\xx}\llbracket\vv/ \Gamma(\xx) \rrbracket$ will occur in $p(\uu,\xx) \cdot \fb_{\yy}$, and then from (C2) we can prove that when constructing $D_{\yy}$, the variables do not exist in $p(\uu,\xx) \cdot \fb_{\yy}$ can be simply abstracted away.

In summary, assuming an invariant template $P(\uu,\xx)$ without fresh variables $\vv$ can greatly simplify the construction of $\mathcal C_{\yy}$, and enables us to generate polynomial invariants for $\mathcal C_{\xx}$.

\begin{figure}
\begin{center}
\includegraphics[width=1.6in,height=1.5in]{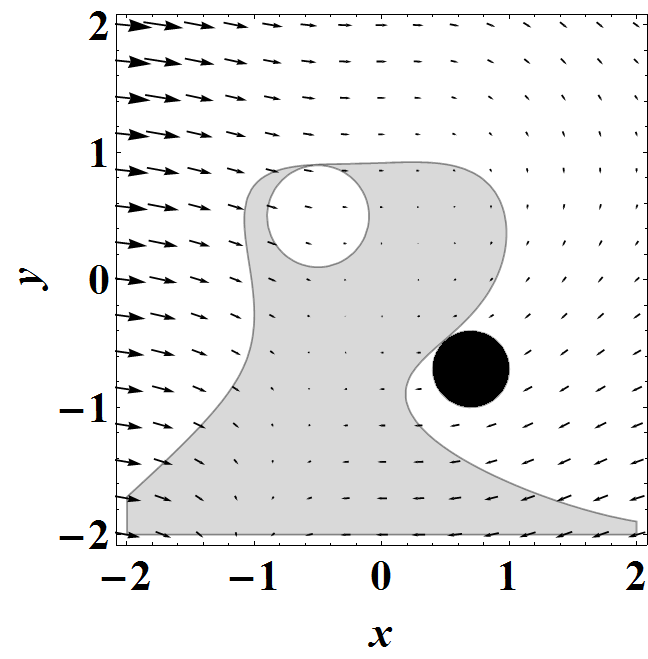}
\hspace{1cm}
\includegraphics[width=1.6in,height=1.5in]{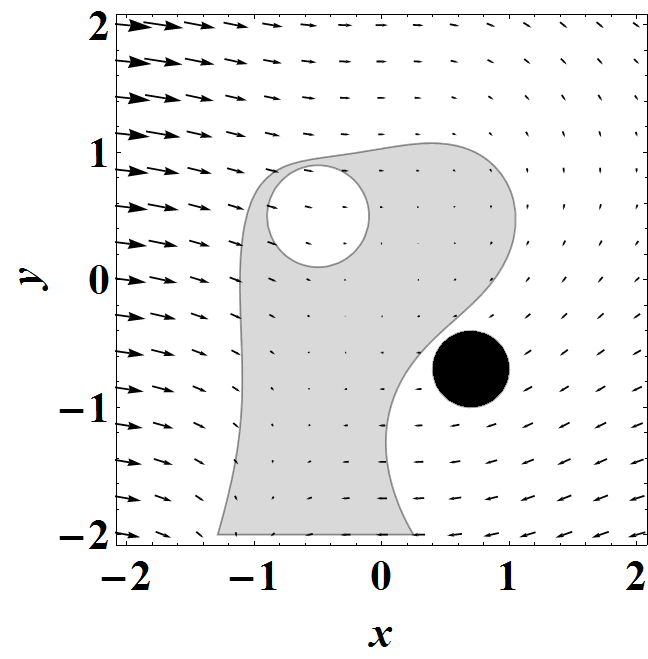}
\caption{Comparison of polynomial and elementary inductive invariants}
\label{fig:inv-trans}
\end{center}

\end{figure}

\begin{example}\label{eg:eds-verify-1}
 Consider the EDS $\mathcal C_{\xx}$ in Example \ref{eg:eds-pds}. We will try to generate a polynomial inductive invariant to verify the safety of $\mathcal C_{\xx}$ w.r.t. an unsafe region $\bar S_{\xx}\,\define\,(x-0.7)^2+(y+0.7)^2-0.09\leq 0$. By the above discussion, the PDS abstraction $\mathcal C_{\yy}$ of $\mathcal C_{\xx}$ can be defined by
 $$\mathcal C_{\yy}\quad\define\quad\big(\Xi_{\xx},\fb_{\xx}\llbracket\vv/\Gamma(\xx)\rrbracket,D_{\xx}\wedge (\ref{eqn:taylor-sin-x})\wedge (\ref{eqn:taylor-exp-x}) \big)$$ with $\vv=\Gamma(\xx)$ given by (\ref{eqn:eg-eds-gamma}). The unsafe region for $\mathcal C_{\yy}$ is $\bar S_{\yy}\,\define\, \bar S_{\xx}$.

 By applying the SOS-relaxation-based invariant generation approach \cite{PJP07,KHSH13} with a polynomial template $p(\uu,x,y)\leq 0$ of degree 5 (in $x,y$) and using the Matlab-based tool YALMIP \cite{Yalmip} and SeDuMi \cite{SeDuMi} (or SDPT3 \cite{Sdpt3}), we successfully generated an invariant that verifies $\neg \bar S_{\xx}$ for $\mathcal C_{\xx}$.
 Please see the left part of Figure~\ref{fig:inv-trans} for an illustration of $\fb_{\xx}$ (the black arrows), $D_{\xx}$ (the outer white box), the synthesized invariant $p(x,y)\leq 0$ (the grey area with curved boundary), $\Xi_{\xx}$ (the white circle inside the invariant) and $\bar S_{\xx}$ (the black circle outside the invariant). The explicit form of $p(x,y)$ is:
 \begin{eqnarray*}
p(x,y)  &:=  & \scriptstyle
-29.5258683+2.7905 x-15.4285 y+7.7870 x^2-20.4040 x y\\
& & \scriptstyle +22.4031 y^2+14.0762 x^3-18.7539 x^2 y+41.8913 x y^2+5.9623 y^3\\
& & \scriptstyle +25.8881 x^4+4.5276 x^3 y+2.6340 x^2 y^2-21.2871 x y^3+5.6462 y^4\\
& & \scriptstyle -9.8303 x^5+0.8716 x^4 y+1.4942 x^3 y^2+9.9083 x^2 y^3-11.0499 x y^4\\
& & \scriptstyle +24.5758 y^5 \qquad .
\end{eqnarray*}

\end{example}

\subsection{Generating Elementary Invariants}
Now we show how to generate elementary invariants for $\mathcal C_{\xx}$ in Example \ref{eg:eds-verify-1}.
\begin{example}\label{eg:eds-verify-2}
  Consider the EDS $\mathcal C_{\xx}$ and unsafe region $\bar S_{\xx}$ in Example \ref{eg:eds-verify-1}. This time we try to generate an inductive invariant for $\mathcal C_{\yy}$ using the template $p(\uu,\xx,\vv)\leq 0$ with all the variables $\vv$ included. According to constraints similar to (C1)-(C3), it requires a more refined abstraction of $\mathcal C_{\xx}$ to reflect the relations between $\xx$ and $\vv$. Here we adopt (W2) for the abstraction of $\Xi_{\xx},D_{\xx}$ and $\bar S_{\xx}$. We define $D_{\yy}$ to be the same one as in Example \ref{eg:eds-pds}. From $\Xi_{\xx}$ it can be deduced that $(x,y)\in B_{\Xi}\,\define\, [-0.9,-0.1]\times [0.1,0.9]$ for any $(x,y)\in\Xi_{\xx}$. Then we can compute the Taylor polynomials of $\vv=\Gamma(\xx)$ over $B_{\Xi}$, and thus get $\Xi_{\yy}$. The abstraction $\bar S_{\yy}$ of $\bar S_{\xx}$ can be obtained similarly. The vector field $\fb_{\yy}$ is given by (\ref{eqn:eg-transode-poly}).

  Using a template $p(\uu,\xx,\vv)\leq 0$ with $p(\uu,\xx,\vv)$ a parametric polynomial of degree 3 (in $\xx,\vv$), we finally obtained an invariant $p(x,y,v_1,v_2,v_3)\leq 0$  that verifies $\neg \bar S_{\yy}$ for $\mathcal C_{\yy}$, which means $p(x,y,\sin{x},e^{-x},\cos{x})\leq 0$ is an invariant of $\mathcal C_{\xx}$ that verifies $\neg \bar S_{\xx}$. The right part of Figure \ref{fig:inv-trans} is an illustration of
  $p(x,y,\sin{x},e^{-x},\cos{x})\leq 0$. The explicit form of $p(x,y,v_1,v_2,v_3)$ is:
  \begin{eqnarray*}
  p & := &  \scriptstyle   -4.955995973+2.6956 x-7.7162 y+1.3633 v_1-1.1243 v_2-1.0806 v_3+0.6966 x^2-8.9155 x y+5.8828 y^2\\
  & & \scriptstyle +3.0691 x v_1-3.6545 y v_1+0.2592 v_1^2-3.3022 x v_2+1.4964 y v_2-0.7498 v_1 v_2-4.2837 v_2^2-0.2079 x v_3\\
  & & \scriptstyle -7.7557 y v_3-1.5121 v_1 v_3+1.5754 v_2 v_3-1.3813 v_3^2-0.0353 x^3-0.3128 x^2 y+0.9184 x y^2+6.6938 y^3\\
  & & \scriptstyle
  -0.1410 x^2 v_1+3.1509 x y v_1+1.8136 y^2 v_1+6.5973 x v_1^2+7.7242 y v_1^2+1.7114 v_1^3-1.0877 x^2 v_2+4.4452 x y v_2\\
  & & \scriptstyle +1.2358 y^2 v_2+1.3919 x v_1 v_2+7.9981 y v_1 v_2+2.5635 v_1^2 v_2-0.8835 x v_2^2+1.4900 y v_2^2+0.0392 v_1 v_2^2+1.1281 v_2^3\\
  & & \scriptstyle -1.8619 x^2 v_3 -2.4300 x y v_3+2.2032 y^2 v_3-2.5384 x v_1 v_3-6.2048 y v_1 v_3-4.9447 v_1^2 v_3+1.6193 x v_2 v_3\\
  & & \scriptstyle +1.2933 y v_2 v_3
 -0.5207 v_1 v_2 v_3-0.2498 v_2^2 v_3+5.8866 x v_3^2+5.1296 y v_3^2+0.7890 v_1 v_3^2+2.0905 v_2 v_3^2-2.3259 v_3^3 \enspace .
\end{eqnarray*}
%
\end{example}

We can see that the elementary invariant is sharper than the polynomial invariant and separates better from the unsafe region. This indicates that by allowing non-polynomial terms in templates, invariants of higher quality may be generated and thus increases the possibility of verifying safety properties of EHSs. Moreover, it also suggests that even for purely polynomial systems, one could assume any kind of elementary terms in a predefined template when generating invariants, which gives a more general method than \cite{RMM12,GJPS14} for generating elementary invariants for PHSs.

%

\subsection{More Experiments}
We have implemented the proposed abstraction approach (not including the part on abstraction of replacement equations) and experimented with it using the following examples on safety verification for EHSs. The formal abstraction algorithms can be found in the appendix, and all the input files for the experiments can be obtained at {\scriptsize \sf {http://lcs.ios.ac.cn/\%7Ezoul/casestudies/fm2015.zip}}

\begin{example}[HIV Transmission]\label{eg:hiv}
  The following continuous dynamics, with the assumption that there is no recruitment of population, has been developed to model HIV transmission \cite{Anderson88}
\begin{equation}\label{eqn:eg-hiv}
\fb \,\define\, \left\{ \begin{array}{lll}
\dot u_1 & = & -\frac{\beta c u_1 u_2} {u_1+u_2+u_3}-\mu u_1\\
\dot u_2 & = & \frac{\beta c u_1 u_2} {u_1+u_2+u_3}-(\mu+\nu)u_2\\
\dot u_3 & = & \nu u_2 -\alpha u_3
\end{array} \right.\enspace ,
\end{equation}
where $u_1(t),u_2(t),u_3(t)$ denote the part of population that is HIV susceptible, HIV infected, and that has AIDS respectively, $\beta$ is the possibility of infection per partner contact, $c$ is the rate of partner change, $\mu$ is the death rate of non-AIDS population, $\alpha$ is the death rate of AIDS patients, and $\nu$ is the rate at which HIV infected people develop AIDS. Note that the dynamics involves non-polynomial term $\frac{1}{u_1+u_2+u_3}$. In this paper, the parameters are chosen to be $\beta=0.2,c=10,\mu=0.008,\alpha=0.95,\nu=0.1$.
We want to verify that with the initial set
$$\Xi\,\define\,u_1\in[9.985,9.995]\wedge u_2\in[0.005,0.015]\wedge u_3\in[0,0.003],$$
the population of AIDS patients alive will always be below 1 (the population is measured in thousands). That is, the system $(\Xi,\fb, D)$ satisfies $S\,\define\, u_3\leq 1$, where $D\,\define\,u_1\geq 0\wedge u_2\geq 0\wedge u_3\geq 0\wedge 0< u_1+u_2+u_3\leq 10.013 \,.\footnote{According to dynamics (\ref{eqn:eg-hiv}), the entire population is non-increasing, so $u_1+u_2+u_3$ has an upper bound.}$
\end{example}

\begin{figure}
\begin{center}
\includegraphics[width=3in,height=1.1in]{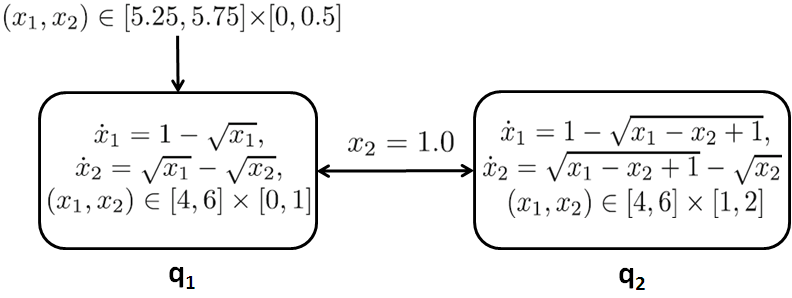}
\caption{HA model of the two-tanks system}\label{fig:tanks}
\end{center}

\end{figure}

\begin{example}[Two-Tanks]\label{eg:tank}
The two-tanks system shown in Figure \ref{fig:tanks} comes from \cite{SKHP97} and has been studied in \cite{RS07,Hydlogic,iSAT-ODE} as a benchmark for safety verification of hybrid systems. It models two connected tanks, the liquid levels of which are denoted by $x_1$ and $x_2$ respectively. The system switches from mode $q_1$ (or $q_2$) to $q_2$ (or $q_1$) when $x_2$ reaches 1 at $q_1$ (or $q_2$). The system's dynamics involve non-polynomial terms such as $\sqrt{x_1}$ or $\sqrt{x_1-x_2+1}$. The verification objective is to show that starting from mode $q_1$ with the initial set
$\Xi_{q_1}\,\define\, 5.25\leq x_1\leq 5.75\wedge 0\leq x_2\leq 0.5$, the system will never reach the unsafe set $\bar S_{q_1}\,\define\,(x_1-4.25)^2+(x_2-0.25)^2-0.0625\leq 0$ when staying at mode $q_1$.
\end{example}

\begin{example}[Lunar Lander]\label{eg:lunar}
Consider a real-world example of the guidance and control of a lunar lander
\cite{ZYZG14}, as illustrated by Figure \ref{fig:hcs}.
The dynamics of the lander is given by
\begin{equation}\label{eqn:eg-hcs}
\fb \,\define\, \left\{ \begin{array}{lll}
\dot v & = & \frac{F_c}{m}-1.622\\
\dot m & = & -\frac{F_c}{2500} \\
\dot F_c & = & 0\\
\dot t & = & 1
\end{array} \right. \enspace ,
\end{equation}
where
$v$ and $m$ denote the vertical velocity and mass of the lunar lander;
$F_c$ denotes the thrust imposed on the lander, which is kept constant during one sampling cycle of length 0.128 seconds; at each sampling point, $F_c$ is updated according to the guidance law shown in the right part of Figure \ref{fig:hcs}.
\begin{figure}[h]
\begin{center}
\includegraphics[width=0.9in,height=1.2in]{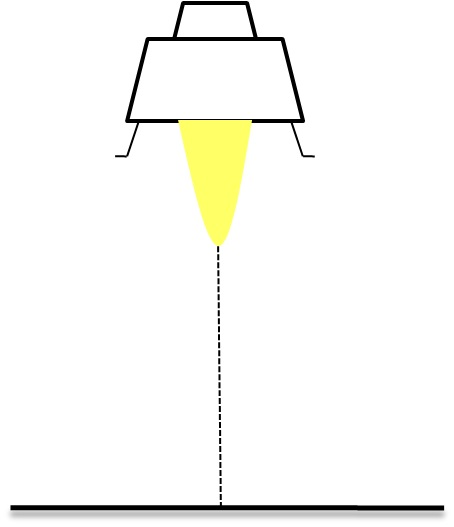}
\hspace{.6cm}
\includegraphics[width=1.8in,height=1.25in]{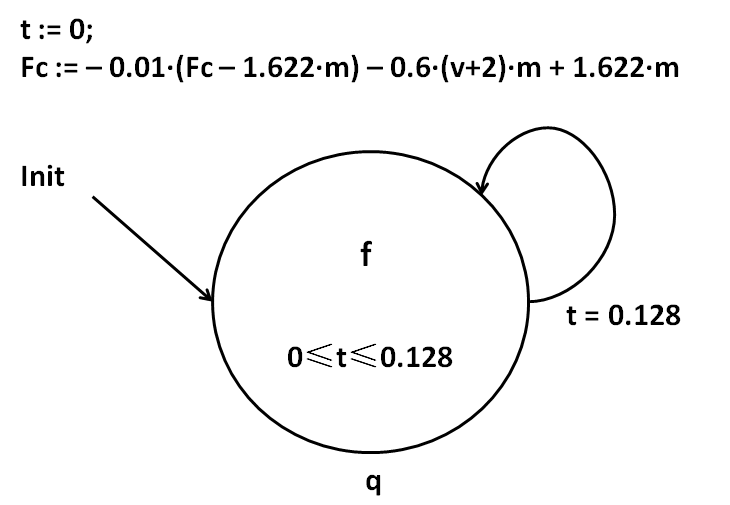}
\caption{The lunar lander and its guidance-control system}
\label{fig:hcs}
\end{center}
\end{figure}Note that the derivative of $v$ involves non-polynomial expression $\frac{1}{m}$.
We want to verify that with the initial condition $t=0$s, $v=-2$m/s, $m=1250$kg, $F_c=2027.5$N, the vertical velocity of the lunar lander will be kept around the target velocity $-$2m/s, i.e. $|v-(-2)|\leq \varepsilon$, where $\varepsilon=0.05$ is the specified bound for fluctuation of $v$.
\end{example}

Using the proposed abstraction method and the SOS-relaxation-based invariant generation method, we have successfully verified all the above 3 examples. The time costs on the platform with Intel Core i5-3470 CPU and 4GB RAM running Windows 7 are shown in Table \ref{tbl:timing}.
\begin{table}[h]
\begin{center}
  \caption{Time costs of invariant generation and safety verification for EHSs}\label{tbl:timing}
\begin{tabular}{|c|c|c|c|c|c|}
  \hline
  example  & E.g. \ref{eg:eds-verify-1} & E.g. \ref{eg:eds-verify-2} & E.g. \ref{eg:hiv} & E.g. \ref{eg:tank} & E.g. \ref{eg:lunar} \\
  \hline
  time cost (s) & 1.324 & 7.994 & 5.186 & 0.977 & 2.645\\
  \hline
\end{tabular}
\end{center}
\end{table}

Besides, we have also compared with the performances of the EHS verification tools HSOLVER \cite{RS07}, Flow$^*$ \cite{CAS13}, dReach \cite{GKC13-dreach} and iSAT-ODE \cite{iSAT-ODE} on these examples.\footnote{Note that since Flow$^*$, dReach and iSAT-ODE can only do BMC, we have assumed a time bound of 20s and 10s resp. for E.g. \ref{eg:eds-verify-1} and \ref{eg:hiv}, and a jump bound of 40 steps and 100 steps resp. for E.g. \ref{eg:tank} and \ref{eg:lunar}.} The results are obtained on the same platform as above except for running Ubuntu Linux 14.04. In Table \ref{tbl:comparison}, time is measured in seconds; $-$ means that the verification fails, either because of abnormal termination due to error inflation, or because of non-termination within reasonable amount of time (several hours).

\begin{table}
\begin{center}
  \caption{Verification results of different methods}\label{tbl:comparison}
\begin{tabular}{|c|c|c|c|c|c|c|}
  \hline
                             & EHS2PHS          & HSOLVER  & Flow$^*$  & dReach  & iSAT-ODE \\
  \hline
  E.g. \ref{eg:eds-verify-1} & 1.324            & 0.723    & $-$       & $-$     & $-$      \\
  E.g. \ref{eg:hiv}          & 5.186            & $-$      & $-$       & $-$     & $-$      \\
  E.g. \ref{eg:tank}         & 0.977            & 0.452    & 76.880    & 21.949  & 0.988    \\
  E.g. \ref{eg:lunar}        & 2.645            & $-$      & 20.238    & $-$     & 63.648   \\
  \hline
\end{tabular}
\end{center}
\end{table}
%

From Table \ref{tbl:comparison} we can see that the time costs of the proposed abstraction approach are all acceptable, whereas there do exist examples that existing approaches cannot solve effectively.

\section{Conclusions}\label{sec:conclusion}
In this paper, we presented an approach to reducing an EHS to a PHS by variable transformation, and established the simulation relation between them, so that safety verification of the EHS can be reduced to that of the corresponding PHS. Thus our work enables all the well-established techniques for PHS verification to be applicable to EHSs. In particular, combined with invariant-based approach to safety verification for PHSs, it provides the possibility of overcoming the limitations of existing EHS verification approaches. Experimental results on real-world examples indicated the effectiveness of our approach.

A possible drawback of the proposed approach is that the SOS-based method may cause an incorrect invariant to be generated due to numerical computation errors. To overcome this, we have verified all the synthesized invariants posteriorly using symbolic computation tools.


\bibliographystyle{splncs03}
\bibliography{Myrefs-fm15}

\newpage
\appendix
\section*{Abstraction Algorithms}

\begin{algorithm}[h]
\caption{Reducing an elementary expression to a polynomial one
  (\textbf{VT}($\textit{expr}$, $\textit{eqs}$))}
\label{alg:subst}
\begin{algorithmic}[1]
\REQUIRE An elementary expression $\textit{expr}$ and a set of equations $\textit{eqs}$ as input
\ENSURE The returned expression is polynomial,
    and equals to the input expression in the context of equations $\textit{eqs}$
\IF{$\textit{expr} = c$ or $\textit{expr} = x$}
  \RETURN $(\textit{expr},\textit{eqs})$;
\ELSIF{$\textit{expr} = \frac{\textit{expr}_1}{\textit{expr}_2}$}
  \STATE{$(\textit{expr}_2,\textit{eqs}) = \textbf{VT}(\textit{expr}_2,eqs)$;\quad
  \textbf{return} $\textbf{VT}(\textit{expr}_1*\textit{new\!Var}$,
    $\textit{eqs}.\textbf{add}(\textit{new\!Var},\frac{1}{\textit{expr}_2})$);}
\ELSIF{$\textit{expr} = \textit{expr}_1^{\frac{\textit{n}_1}{\textit{n}_2}}$}
  \STATE{$(\textit{expr}_1,\textit{eqs}) = \textbf{VT}(\textit{expr}_1,eqs)$;\quad
  \textbf{return} ({$\textit{new\!Var}^{\textit{n}_1}$}, $\textit{eqs}.\textbf{add}(\textit{new\!Var},\textit{expr}_1
    ^{\frac{1}{\textit{n}_2}})$); }
\ELSIF{$\textit{expr} = e^{\textit{expr}_1}$}
  \STATE{$(\textit{expr}_1,\textit{eqs}) = \textbf{VT}(\textit{expr}_1,eqs)$;\quad
  \textbf{return} ($\textit{new\!Var}$, $\textit{eqs}.\textbf{add}(\textit{new\!Var},e^{\textit{expr}_1})$); }
\ELSIF{$\textit{expr} = \ln(\textit{expr}_1)$}
  \STATE{$(\textit{expr}_1,\textit{eqs}) = \textbf{VT}(\textit{expr}_1,eqs)$;\quad
  \textbf{return} ($\textit{new\!Var}$, $\textit{eqs}.\textbf{add}(\textit{new\!Var},
    \ln(\textit{expr}_1))$); }
\ELSIF{$\textit{expr} = \sin(\textit{expr}_1)$}
  \STATE{$(\textit{expr}_1,\textit{eqs}) = \textbf{VT}(\textit{expr}_1,eqs)$;\quad
  \textbf{return} ($\textit{new\!Var}$, $\textit{eqs}.\textbf{add}(\textit{new\!Var},
    \sin(\textit{expr}_1))$); }
\ELSIF{$\textit{expr} = \cos(\textit{expr}_1)$}
  \STATE{$(\textit{expr}_1,\textit{eqs}) = \textbf{VT}(\textit{expr}_1,eqs)$;\quad
  \textbf{return} ($\textit{new\!Var}$, $\textit{eqs}.\textbf{add}(\textit{new\!Var},
    \cos(\textit{expr}_1))$); }
\ELSIF {$\textit{expr} = \textit{expr}_1 + \textit{expr}_1$}
  \STATE{$(\textit{expr}_1,\textit{eqs}) = \textbf{VT}(\textit{expr}_1,eqs)$;\quad
  $(\textit{expr}_2,\textit{eqs}) = \textbf{VT}(\textit{expr}_2,eqs)$;}
   \RETURN ($\textit{expr}_1 + \textit{expr}_2$,$\textit{eqs}$); \quad
\ELSIF {$\textit{expr} = \textit{expr}_1 - \textit{expr}_1$}
  \STATE{$(\textit{expr}_1,\textit{eqs}) = \textbf{VT}(\textit{expr}_1,eqs)$;\quad
  $(\textit{expr}_2,\textit{eqs}) = \textbf{VT}(\textit{expr}_2,eqs)$;}
   \RETURN ($\textit{expr}_1 - \textit{expr}_2$,$\textit{eqs}$); \quad
\ELSE {
  \STATE{$(\textit{expr}_1,\textit{eqs}) = \textbf{VT}(\textit{expr}_1,eqs)$;\quad
  $(\textit{expr}_2,\textit{eqs}) = \textbf{VT}(\textit{expr}_2,eqs)$;}
  \RETURN ($\textit{expr}_1 \times \textit{expr}_2$, $\textit{eqs}$);}
\ENDIF
\end{algorithmic}
\end{algorithm}

In Algorithm \ref{alg:subst},
  $\textit{new\!Var}$ denotes  a fresh variable, and $\textit{eqs}$
  records the replacements during the variable transformation.

\begin{algorithm}[h]
\caption{Updating the dynamical system according to the replacement equations $\textit{eqs}$ (\textbf{U}($\textit{odes}$, $\textit{eqs}$))}
\label{alg:updateode}
\begin{algorithmic}[1]
\REQUIRE Polynomial differential equations $odes$ and a set of equations $\textit{eqs}$ as input (where all expressions in $\textit{eqs}$ are polynomial except the outermost operator)
\ENSURE The resulting polynomial differential equations simulate the initial $odes$ and $\textit{eqs}$
\FOR{$(\textit{var},\textit{expr})$ \textbf{in} $\textit{eqs}$}
  \IF{$\textit{expr} = \frac{1}{\textit{expr}_2}$}
    \STATE{$odes.\textbf{add}(\textit{var},-\textit{var}^2*\dot{\textit{expr}_2})$;}
  \ELSIF{$\textit{expr} = \textit{expr}_1^{\frac{1}{\textit{n}_2}}$}
    \STATE{$\textit{eqs}.\textbf{add}(\textit{new\!Var},1/\textit{expr})$; \quad
     $odes.\textbf{add}(\textit{var},\frac{1}{\textit{n}_2}*\textit{new\!Var}^{\textit{n}_2-1}*\dot{\textit{expr}_1})$;}
  \ELSIF{$\textit{expr} = e^{\textit{expr}_1}$}
    \STATE{$odes.\textbf{add}(\textit{var},\textit{var}*\dot{\textit{expr}_1})$;}
  \ELSIF{$\textit{expr} = \ln(\textit{expr}_1)$}
    \STATE{$\textit{eqs}.\textbf{add}(\textit{new\!Var},\frac{1}{\textit{expr}_1})$; \quad
     $odes.\textbf{add}(\textit{var},\textit{new\!Var}*\dot{\textit{expr}_1})$;}
  \ELSIF{$\textit{expr} = \sin(\textit{expr}_1)$}
    \STATE{$\textit{eqs}.\textbf{add}(\textit{new\!Var},\cos(\textit{expr}_1))$; \quad
     $odes.\textbf{add}(\textit{var},\textit{new\!Var}*\dot{\textit{expr}_1})$;}
  \ELSIF{$\textit{expr} = \cos(\textit{expr}_1)$}
    \STATE{$\textit{eqs}.\textbf{add}(\textit{new\!Var},\sin(\textit{expr}_1))$; \quad
    $odes.\textbf{add}(\textit{var},-\textit{new\!Var}*\dot{\textit{expr}_1})$;}
  \ELSE
    \STATE{The algorithm should not run this branch;}
  \ENDIF
\ENDFOR
\RETURN $odes$;
\end{algorithmic}
\end{algorithm}

In Algorithm \ref{alg:updateode},
  $\textbf{op}$, $\textbf{left}$, and $\textbf{right}$ returns the outermost operation,
  and its left and right operands for a given expression, respectively.
    $\textbf{left}$ and $\textbf{right}$ return the operand in case the outmost operation is one ary;
  $\textit{new\!Var}$ denotes a fresh variable. 
Algorithm \ref{alg:updateode} must terminate,
  because the number of elements of $\textit{eqs}$ can only increase finite times, obviously, no more than the number of the subexpressions of the EDS.

\begin{algorithm}[h]
\caption{Transforming elementary ODEs to polynomial ODEs (\textbf{TransEODEs}($\textit{odes}$, $\textit{eqs}$))}
\label{alg:transodes}
\begin{algorithmic}[1]
\REQUIRE ODEs $\textit{odes}$ and a list of replacement equations $\textit{eqs}$ as input
\ENSURE The resulting ODEs are polynomial
\FOR {$ode$ \textbf{in} $odes$}
  \FOR{$exp$ \textbf{in} $\textbf{omExp}(ode)$}
    \STATE{($\textit{exp}$, $\textit{eqs}$) =
      \textbf{VT}($\textit{exp}$, $\textit{eqs}$);}
  \ENDFOR
\ENDFOR
\STATE{\textbf{return} \textbf{U}($\textit{odes}$, $\textit{eqs}$);}
\end{algorithmic}
\end{algorithm}

In algorithm \ref{alg:transodes}, $\textbf{omExp}(ode)$ returns the set of  the outmost expressions of $ode$, and \textbf{VT} and \textbf{U} call Algorithm 1 and 2, respectively.

\begin{algorithm}[h]
\caption{Transforming elementary hybrid systems (\textbf{TransEHS}($\textit{hs}$))}
\label{alg:transhs}
\begin{algorithmic}[1]
\REQUIRE An elementary hybrid system $\textit{hs}$ as input
\ENSURE The resulting hybrid system is a PHS which simulates the input EHS
\STATE{Set $eqs$ to empty;}
\FOR {$mode$ \textbf{in} $hs$}
  \FOR {$exp$ \textbf{in} $\textbf{omExp}(mode.init)$}
    \STATE{($exp, eqs$) = \textbf{VT}($exp, eqs$);}
  \ENDFOR
  \FOR {$exp$ \textbf{in} $\textbf{omExp}(mode.domain)$}
    \STATE{($exp, eqs$) = \textbf{VT}($exp, eqs$);}
  \ENDFOR
  \FOR {$exp$ \textbf{in} $\textbf{omExp}(mode.guard)$}
    \STATE{($exp, eqs$) = \textbf{VT}($exp, eqs$);}
  \ENDFOR
  \FOR {$exp$ \textbf{in} $\textbf{omExp}(mode.reset)$}
    \STATE{($exp, eqs$) = \textbf{VT}($exp, eqs.expr$);}
  \ENDFOR
  \STATE{mode.odes = \textbf{TransEODEs}($\textit{odes}$, $\textit{eqs}$);}
\ENDFOR
\STATE{\textbf{return} $\textit{hs}$;}
\end{algorithmic}
\end{algorithm}

In Algorithm \ref{alg:transhs}, $\textbf{omExp}(\textit{form})$ returns the set of the outmost expressions of formula $\textit{form}$, and \textbf{VT} and \textbf{TransODEs} call Algorithm 1 and 3, respectively.

\end{document}